\documentclass[aos]{imsart}

\usepackage[utf8]{inputenc}
\usepackage[T1]{fontenc}
\usepackage{graphicx}
\usepackage{longtable,booktabs}
\usepackage{wrapfig}
\usepackage{rotating}
\usepackage[normalem]{ulem}
\usepackage{amsmath,amsthm}
\usepackage{amssymb}

\RequirePackage[round]{natbib}
\newcommand*{\textcite}{\citet}
\newcommand*{\parencite}{\citep}

\usepackage[colorlinks,citecolor=blue,urlcolor=blue,linkcolor=blue,breaklinks=true]{hyperref}
\usepackage{pdflscape}
\usepackage{amsmath,mathtools,mathrsfs,bm}
\usepackage{enumerate}
\usepackage{tabularx,array}
\usepackage[capitalise]{cleveref}

\usepackage{cancel}
\usepackage{subcaption}

\usepackage{algorithm}
\usepackage{algpseudocode}

\usepackage{tikz}
\usetikzlibrary{calc,cd,decorations.pathmorphing,shapes}

\usepackage{mathtools}

\definecolor{blue-violet}{rgb}{0.54, 0.17, 0.89}
\definecolor{antiquefuchsia}{rgb}{0.57, 0.36, 0.51}
\definecolor{amethyst}{rgb}{0.6, 0.4, 0.8}
\definecolor{blue-violet}{rgb}{0.54, 0.17, 0.89}
\definecolor{ao}{rgb}{0.0, 0.5, 0.0}
\definecolor{blue(ncs)}{rgb}{0.0, 0.53, 0.74}
\definecolor{dgreen}{rgb}{0.12, 0.3, 0.17}
\definecolor{cadmiumgreen}{rgb}{0.0, 0.42, 0.24}
\definecolor{darkolivegreen}{rgb}{0.33, 0.42, 0.18}
\definecolor{dartmouthgreen}{rgb}{0.05, 0.5, 0.06}

\newcommand{\tild}{\raise.17ex\hbox{ $\scriptstyle\sim$ }}

\newcommand{\indep}{\mathrel{\text{\scalebox{1.07}{$\perp\mkern-10mu\perp$}}}}

\usepackage{soul}

\theoremstyle{plain}
\newtheorem{theorem}{Theorem}

\newtheorem{lemma}{Lemma}
\newtheorem{corollary}{Corollary}
\theoremstyle{remark}
\newtheorem{definition}{Definition}
\newtheorem{remark}{Remark}

\usepackage{adjustbox}
\usepackage{centernot}
\usepackage{tikz}
\usetikzlibrary{calc,cd,decorations.pathmorphing,shapes}
\usepackage[mathscr]{euscript}

\DeclareMathOperator{\gG}{\mathscr{G}}
\DeclareMathOperator{\gF}{\mathscr{F}}

\DeclareMathOperator{\sV}{\mathrm{V}}
\DeclareMathOperator{\sU}{\mathrm{U}}

\DeclareMathOperator{\sD}{\mathcal{D}}
\DeclareMathOperator{\sB}{\mathcal{B}}

\DeclareMathOperator{\sP}{\mathcal{P}}
\DeclareMathOperator{\sW}{\mathcal{W}}

\DeclareMathOperator{\an}{\mathrm{an}}
\DeclareMathOperator{\de}{\mathrm{de}}
\DeclareMathOperator{\pa}{\mathrm{pa}}
\DeclareMathOperator{\ch}{\mathrm{ch}}

\usepackage{graph-separation}

\newcommand{\barrow}{\begin{tikzpicture}[scale=.3, baseline=-3mm, transform shape, >=stealth]
\node (a) at (0, -0.6) {};
\node (b) at (2.5, -0.6) {};
\draw [<->,thick] (a) -- (b);
\end{tikzpicture}}

\newcommand{\dbarrow}{\begin{tikzpicture}[scale=.3, baseline=-3mm, transform shape, >=stealth]
\node (a) at (0, -0.6) {};
\node (b) at (2.5, -0.6) {};
\draw [<->,dashed,thick] (a) -- (b);
\end{tikzpicture}}

\renewcommand{\rightsquigarrow}{\nosquigfull}
\renewcommand{\leftrightsquigarrow}{\fullsquigfull}
\newcommand{\leftsquigarrow}{\fullsquigno}

\let \tilde \widetilde

\begin{document}

\begin{frontmatter}
\title{Confounder Selection via Iterative Graph Expansion}
\runtitle{Confounder Selection via Iterative Graph Expansion}

\begin{aug}
\author[A]{\fnms{F.~Richard}~\snm{Guo}\ead[label=e1]{ricguo@umich.edu}
\orcid{0000-0002-2081-7398}}
\and
\author[B]{\fnms{Qingyuan}~\snm{Zhao}\ead[label=e2]{qyzhao@statslab.cam.ac.uk}}
\address[A]{Department of Statistics, University of Michigan\printead[presep={,\ }]{e1}}
\address[B]{Statistical Laboratory, University of Cambridge\printead[presep={,\ }]{e2}}

\end{aug}

\begin{abstract}
Confounder selection, namely choosing a set of covariates to control
for confounding between a treatment and an outcome, is arguably the
most important step in the design of an observational study. Previous
methods, such as Pearl's back-door criterion, typically
require pre-specifying a causal graph, which can often be difficult in
practice. We propose an interactive procedure for confounder selection
that does not require pre-specifying the graph or the set of observed
variables. This procedure iteratively expands the causal graph by
finding what we call ``primary adjustment sets'' for a pair of
possibly confounded variables. This can be viewed as inverting a
sequence of marginalizations of the underlying causal
graph. Structural information in the form of primary adjustment sets
is elicited from the user, bit by bit, until either a set of
covariates is found to control for confounding or it can be
determined that no such set exists.
Other information, such as the causal relations between confounders,
is not required by the procedure.
We show that if the user correctly
specifies the primary adjustment sets in every step, our procedure is
both sound and complete.
\end{abstract}

\begin{keyword}[class=MSC]
\kwd[Primary ]{62A09}
\kwd[; secondary ]{62D20}
\end{keyword}

\begin{keyword}
\kwd{Causal inference}
\kwd{observational study}
\kwd{back-door criterion}
\kwd{background knowledge}
\kwd{adjustment}
\kwd{marginalization}
\end{keyword}

\end{frontmatter}

\section{Introduction}
\label{sec:introduction}
Consider an observational study where the causal effect of a treatment
variable $X$ on an outcome variable $Y$ is of interest. Arguably,
the single most widely used strategy for identifying the causal effect
is through \emph{confounder adjustment}, which employs a set of observed
covariates that are carefully chosen to control for confounding. Let
$Y(x)$ be the potential outcome of $Y$ had the treatment $X$ been
intervened on and set to level $x$. A set of covariates $S$ controls
for confounding if for every $x$ in the support of $X$, $X$ and $Y(x)$
are independent within each stratum defined by $S$, a condition known
as conditional ignorability or conditional exchangeability
\citep{rosenbaumCentralRolePropensity1983,greenland2009identifiability,hernan2020causal}. The
task of choosing such a set of covariates is called \emph{confounder
  selection}.

While there exist a variety of approaches and criteria for confounder
selection (see \citealp{guo2022confounder} for a recent survey), it is
clear that this cannot be answered by data alone and hence is
fundamentally different from statistical variable selection (e.g., stepwise algorithms in linear regression). That is,
domain knowledge about the underlying causal mechanism or structure is
indispensable. Causal
graphical models provide an intuitive framework for formalizing such
knowledge. Specifically, if we have available a causal directed
acyclic graph (DAG) model for relevant variables in an observational
study, the back-door criterion \citep{pearlBayesianAnalysisExpert1993}
answers whether or not a set of covariates $S$ controls for confounding, and this criterion is
complete in a sense that will be described in
\cref{sec:latent-projection}. The set $S$ is called a \emph{sufficient
  adjustment set} if it satisfies the back-door criterion.
When there is more than one sufficient adjustment
set, one may wish to further choose a set based on its size or
statistical efficiency \citep[see,
e.g.,][]{henckel2022,rotnitzky2020efficient,smucler2022note}. Yet,
these are \emph{secondary objectives} that we will largely set aside for the rest of this
paper (except for \cref{sec:choice-graph-expans}). As far as validity is concerned, confounder selection is sometimes
considered a ``solved problem'' in light of the back-door
criterion, provided that a causal DAG (or a latent projection of
the DAG) can be pre-specified to represent our domain knowledge and
assumptions.

However, the back-door criterion is often difficult to apply in
practice. As the back-door criterion is a global condition about
the candidate set of covariates $S$, the treatment $X$, the outcome
$Y$ and other
variables in the system (see \cref{sec:latent-projection}), a
practitioner must be able to (1) conceive
all the variables, observed or unobserved, that are relevant, (2)
posit all causal relations among these variables, (3) understand how a
DAG encodes causal assumptions, and (4) draw
the DAG accordingly, or at least a large portion of it, to transcribe
the posited relations graphically. While tools and protocols for drawing DAGs
have been developed to some extent \citep{shrier2008reducing,haber2022dagwood}, this
is still a formidable process in practice. It is often difficult to
conceive all the relevant variables, let alone posit all causal
relations among them.

\subsection{Overview of the iterative graph expansion procedure}
\label{sec:overv-iter-graph}

In this paper, we take an interactive, bottom-up approach to
confounder selection that does
not require pre-specifying the causal graph. Our method, called
\emph{iterative graph expansion}, is
based on a symmetric reformulation of Pearl's back-door criterion and
can be viewed as the inverse of \emph{graph marginalization} or
\emph{latent projection}
\citep{verma1990equivalence}. The knowledge about the underlying
causal graph is elicited from the user, bit by bit, until one or more
sets of covariates that meet the symmetric back-door criterion are found,
or it is determined that no such set exists. For this procedure, a key new concept is a
\emph{primary adjustment set}: a candidate adjustment set for a pair of variables is called primary
if for every common ancestor of the two variables, at least one of the
two causal paths from the ancestor to these two variables is blocked
by the adjustment set. Intuitively, a primary adjustment set removes
all ``immediate'' confounding between the two variables.

More specifically, the process of graph expansion starts with a working
graph consisting of two vertices --- the treatment $X$ and the outcome $Y$ --- and a dashed
bidirected edge, representing possible uncontrolled
confounding, between them. In each step, the user is asked to provide
candidates of \emph{primary adjustment sets} to expand a dashed
bidirected edge selected from the current working graph; if no such set exists,
the dashed bidirected edge is changed to a solid edge. If primary
adjustment sets do exist, then every such set leads to an expanded graph:
the selected dashed edge is removed and those vertices in the primary
adjustment set are introduced to the graph, with
a dashed bidirected edge drawn between every new vertex and every old
vertex, as well as between every pair of new vertices.
This process is repeated until the treatment and the outcome are no longer
connected by solid or dashed bidirected edges. When this occurs, variables other than $X$
and $Y$ in the working graph form a sufficient adjustment set.

\begin{figure}[tb]
\centering
\begin{tikzpicture}
\tikzset{rv/.style={circle,inner sep=1pt,draw,font=\sffamily},
lv/.style={circle,inner sep=1pt,fill=gray!20,draw,font=\sffamily},
fv/.style={rectangle,inner sep=1.5pt,fill=gray!20,draw,font=\sffamily},
node distance=12mm, >=stealth}
\node[rv] (X) {$X$};
\node[rv, above right of=X] (B) {$B$};
\node[rv, below right of=B] (Y) {$Y$};
\node[rv, above left of=B] (D) {$C$};
\node[rv, above right of=B] (C) {$D$};
\draw[->, thick] (B) -- (X);
\draw[->, thick] (B) -- (Y);
\draw[->, thick] (D) -- (B);
\draw[->, thick] (D) -- (X);
\draw[->, thick] (C) -- (B);
\draw[->, thick] (C) -- (Y);
\end{tikzpicture}
\quad
\begin{tikzpicture}
\tikzset{rv/.style={circle,inner sep=1pt,draw,font=\sffamily},
node distance=15mm, >=stealth}
\node[rv] (X) {$X$};
\node[rv, right of=X, xshift=0.5cm] (Y) {$Y$};
\draw[<->, dashed, thick, color=red] (X) -- (Y);
\end{tikzpicture}
\quad
\begin{tikzpicture}
\tikzset{rv/.style={circle,inner sep=1pt,draw,font=\sffamily},
lv/.style={circle,inner sep=1pt,fill=gray!20,draw,font=\sffamily},
fv/.style={rectangle,inner sep=1.5pt,fill=gray!20,draw,font=\sffamily},
node distance=15mm, >=stealth}
\node[rv] (X) {$X$};
\node[rv, above right of=X] (B) {$B$};
\node[rv, below right of=B] (Y) {$Y$};
\draw[<->, thick, dashed] (B) -- (X);
\draw[<->, thick, dashed, red] (B) -- (Y);
\end{tikzpicture}
\quad
\begin{tikzpicture}
\tikzset{rv/.style={circle,inner sep=1pt,draw,font=\sffamily},
lv/.style={circle,inner sep=1pt,fill=gray!20,draw,font=\sffamily},
fv/.style={rectangle,inner sep=1.5pt,fill=gray!20,draw,font=\sffamily},
node distance=15mm, >=stealth}
\node[rv] (X) {$X$};
\node[rv, above right of=X] (B) {$B$};
\node[rv, below right of=B] (Y) {$Y$};
\node[rv, above right of=B] (D) {$D$};
\draw[<->, thick, dashed] (X) -- (B);
\draw[<->, thick, dashed] (D) -- (B);
\draw[<->, thick, dashed, bend right] (D) to (X);
\draw[<->, thick, dashed, red] (D) to (Y);
\end{tikzpicture}
\quad
\begin{tikzpicture}
\tikzset{rv/.style={circle,inner sep=1pt,draw,font=\sffamily},
lv/.style={circle,inner sep=1pt,fill=gray!20,draw,font=\sffamily},
fv/.style={rectangle,inner sep=1.5pt,fill=gray!20,draw,font=\sffamily},
node distance=15mm, >=stealth}
\node[rv] (X) {$X$};
\node[rv, above right of=X] (B) {$B$};
\node[rv, below right of=B] (Y) {$Y$};
\node[rv, above right of=B] (D) {$D$};
\draw[<->, thick, dashed] (X) -- (B);
\draw[<->, thick, dashed] (D) -- (B);
\draw[<->, thick, dashed, bend right] (D) to (X);
\end{tikzpicture}

\caption{An illustration of the iterative graph expansion using
  the ``butterfly bias'' example \citep{ding2015adjust}. The edge chosen for expansion is
  highlighted in red. }
\label{fig:butterfly}
\end{figure}

\Cref{fig:butterfly} illustrates the iterative graph
expansion when the underlying causal graph (the leftmost
graph) is a ``butterfly'' \citep{ding2015adjust}. In the first iteration, the algorithm expands the potential bidirected edge $X \dbarrow Y$ with primary adjustment set $\{B\}$.
This removes $X \dbarrow Y$, but at the same time creates two more potential bidirected edges, $B
\dbarrow X$ and $B \dbarrow Y$. The second iteration of the algorithm
expands $B \dbarrow Y$ by further adding $\{D\}$ to the
graph, which creates three more potential bidirected edges,
$D \dbarrow X$, $D \dbarrow Y$, and $D \dbarrow B$. The next iteration
simply removes $D \dbarrow Y$, as there is no
confounding between $D$ and $Y$. This leads to the rightmost graph,
where $X$ and $Y$ are not connected by bidirected edges. Hence, the algorithm
returns $\{B,D\}$ as a sufficient adjustment set.
Similarly, the set $\{B,C\}$ can also be identified as a sufficient
adjustment set through another sequence of expansions.
In fact, $\{B,C\}$ and $\{B,D\}$ are the only two minimal sufficient
adjustment sets in this example, and $\{B,D\}$ is the so-called efficient
adjustment set \citep{rotnitzky2020efficient,henckel2022}. More
examples are given in \Cref{sec:an-example} and \Cref{sec:unobserved}.

Compared to existing methods for confounder selection, such as directly
applying the back-door criterion, the iterative graph expansion
procedure has several advantages.
First of all, it makes structural queries ``economically'' in the sense
that only so much
information needed for confounder selection is elicited from the user.
All other information, including the presence or absence of directed
edges among the confounders, is never requested.
Second, familiarity with causal graphs (such
as how to apply d-separation) is not required to deploy the procedure;
the user is only expected to identify common causes and mediators.
Finally, we show that this procedure is sound and complete in the
following sense: if all the primary adjustment
sets specified by the user are indeed primary, then every
adjustment set identified by the procedure is sufficient; further, if
the user correctly specifies all the minimal primary adjustment sets
in each step, then every minimal sufficient adjustment set will be
identified by the procedure. To facilitate the use of our procedure,
we provide a Shiny web application accessible
from \url{https://ricguo.shinyapps.io/InteractiveConfSel/}.

It is worth mentioning that there are other proposals in the
literature for confounder selection that only require partial
knowledge about the causal graph, including, notably, the \emph{disjunctive
  cause criterion} due to \citet{vanderweele2011new}; see also
\citet{vanderweele2019principles} for its variants. The disjunctive
cause criterion selects all the observed pre-treatment covariates that
are ancestors (i.e.\ direct or indirect causes) of the treatment, the
outcome or both. \citet{vanderweele2011new} showed that this
adjustment set is sufficient whenever the set of observed
pre-treatment covariates contains any sufficient adjustment set as a
subset; see also \citet{richardson2018discussion,guo2022confounder}. While
this criterion can be useful when data has already been collected and domain
knowledge is scarce, verifying the assumption that the set of
collected covariates contains a sufficient adjustment set can be as
difficult as the task of confounder selection itself.
Moreover, the disjunctive cause criterion and many other
existing proposals require pre-specifying the set of observed
covariates. As such, they are not best suited for designing an
observational study, where the primary goal is to determine what data
need to be collected. Finally, to the best of our knowledge, all the
existing methods for confounder selection are ``static'' --- they
cannot learn from the user in an interactive way.

\subsection{Other contributions and organization of this paper}

The iterative graph expansion procedure is built on several
conceptual and technical novelties, summarized as follows.

First, we provide a definition of \emph{confounding path},
the symmetric structure that induces confounding between two
variables. Specifically, a path between $A$ to $B$ is called
\emph{confounding} if it has two endpoint arrowheads. This allows us
to symmetrize Pearl's back-door criterion
and reduce the task of confounder selection to blocking all confounding paths
between the treatment and the outcome.
The notion of a confounding \emph{path} complements earlier efforts in
the literature towards properly defining a \emph{confounder}, or more
precisely, a confounding
\emph{variable}. \citet{vanderweeleDefinitionConfounder2013} reviewed
several popular notions of confounders and showed that none
of them is satisfactory. They further provided an alternative
definition: a confounder is a pre-exposure
covariate $C$ for which there exists a set of other covariates $S$
such that $(S,C)$ is a sufficient adjustment set but no proper subset
of $(S,C)$ is sufficient. However, this definition does not lead to a practical
procedure for selecting confounders.

Second, we adopt the notion of \emph{confounding arcs} from \citet[Section
3.5]{pearl2009} and provide a concrete definition: a confounding arc
is a confounding path with no colliders.
It is easy to see that every confounding path can be
decomposed into one or more confounding arcs. Importantly, when a set
of covariates block a confounding arc, all of its supersets also block
the same arc. Yet, the same property does not hold for confounding
paths consisting of more than one confounding arc. Nevertheless, we
show that, by iteratively blocking confounding arcs, we can eventually
block all the confounding paths between the treatment and the outcome.

Third, we introduce a \emph{district criterion} for confounder
selection, which posits that a set $S$ is an adjustment set if $X$ and
$Y$ are in different ``districts'' in the marginal graph over $S \cup
\{X,Y\}$. This is equivalent to the symmetric back-door criterion but
can be easier to check in practice.

Finally, to develop our procedure and prove its soundness and
completeness, we introduce a set of refined m-connection/separation
relations, along with a set of notation, to reason about confounding
paths and arcs. As with the usual m-connection \citep{richardson03_markov_proper_acycl_direc_mixed_graph}, we show that these
relations are preserved by marginalization.

The rest of this paper is organized as follows. In
\cref{sec:graph-preliminaries}, we introduce some preliminaries on
causal graphical models, including some basic graphical terms,
m-separation and marginalization. In \cref{sec:suff-adjustm-sets},
we introduce confounding paths and confounding arcs, with which we
formulate a symmetric version of the back-door criterion. We study the
properties of some refined m-connections in \cref{sec:m-conn}. We show
that marginalization preserves connections by confounding paths
(\cref{thm:preserve}), which then leads to the district criterion
(\cref{sec:district-criterion}). We present our iterative graph expansion procedure in \cref{sec:expand}. When the user acts
like an oracle and answers all queries about primary adjustment sets
correctly, our procedure is
shown to be both sound and complete (\cref{thm:alg}). Practical
considerations and subroutines needed for the procedure are studied
in \cref{sec:choice-graph-expans,sec:find-primary}. We conclude with
some discussion in \cref{sec:discuss}. A more complex example is
provided in \Cref{sec:an-example} to illustrate the iterative graph
expansion procedure. Proofs and other supplementary materials can be
found in the appendices.

\section{Preliminaries of causal graphical models}
\label{sec:graph-preliminaries}

\subsection{Basic graphical concepts}
\label{sec:basic-graph-conc}

A \emph{directed mixed graph} $\gG = (\sV, \sD, \sB)$ is a graph over a finite
vertex set $\sV$ that consists of two types of edges, directed edges $\sD
\subseteq \sV \times \sV$ and bidirected edges $\sB \subseteq \sV
\times \sV$. Each vertex in $\sV$ represents a random variable.
We write $A \rdedge B$ for a directed edge $(A,B) \in \sD$ and $A \bdedge B$ for a bidirected edge $(A,B) \in \sB$ (and hence
also $(B,A) \in \sB$). A maximal set of vertices connected by
bidirected edges is called a \emph{district}.
Between two vertices $A,B \in \sV$, we allow the existence of both a directed edge and a bidirected edge.
A \emph{directed cycle} is a sequence of directed edges $v_1 \rdedge \dots \rdedge v_l \rdedge v_1$ with $l > 1$.
If $\gG$ contains no directed cycle, we call $\gG$ an \emph{acyclic directed mixed graph} (ADMG).
Further, if $\gG$ contains no bidirected edge, we call it a \emph{directed acyclic graph} (DAG).

A \emph{walk} is a sequence of adjacent edges of any type or orientation.
If all vertices on a walk are distinct, we say it is a \emph{path}.
For ADMGs, it is necessary
to specify a walk or path as a sequence of edges instead of vertices, as
there may exist both directed and bidirected edges between two
vertices. A walk or path is directed if it is formed of directed edges
pointing in the same direction.

For any walk $\pi$, a non-endpoint vertex $A$ is said to be a
\emph{collider} on $\pi$ if both the edges before and after $A$ have
an arrowhead pointing to $A$, or in other words, if $\pi$ contains
$\rdedge A \ldedge$, $\rdedge A \bdedge$,
$\bdedge A \ldedge$ or $\bdedge A
\bdedge$.
Note that a vertex can be a collider on one path and a non-collider on
another path.  We use a squiggly line ($\nosquigno$) to signify that a
given walk or path contains no collider and we call such a walk or
path an \emph{arc}; some authors use \emph{trek} to
describe the same concept
\citep{spirtesCausationPredictionSearch1993,sullivantTrekSeparationGaussian2010}.
An arc can be further specified by the arrowhead or tail added
to the squiggly line.
A directed path from $A$ to $B$ is
written as $A \rdpath B$ or $B \ldpath A$. A path from $A$ to $B$ with no
colliders and two endpoint arrowheads is called a \emph{confounding arc}
and written as $A \confarc B$. We use a half-arrowhead to indicate that the
endpoint can be either an arrowhead or a tail. For example, $A \halfsquigfull B$
means is the path is either $A \rdpath B$ or $A \confarc B$.

We adopt the common familial terminologies for graphical models. A
vertex $A$ is said to be a \emph{parent} of another vertex $B$ and
$B$ a \emph{child} of $A$ in the
ADMG $\gG$ if $\gG$ contains $A \rdedge B$; and $A$ is said to be
an \emph{ancestor} of $B$ and $B$ a \emph{descendant} of $A$ if $\gG$
contains $A \rdpath B$.
(This differs from the convention
that $A$ is considered both an ancestor and a descendant of itself used by
many authors.)
With these relations, we define sets $\pa_{\gG}(v), \ch_{\gG}(v),
\an_{\gG}(v), \de_{\gG}(v)$ for a vertex $v \in \sV$ and extend these
definitions to a vertex set $A \subseteq \sV$ disjunctively, e.g.,
$\pa_{\gG}(A) := \cup_{v \in A} \pa_{\gG}(v)$, $ \ch_{\gG}(A) :=
\cup_{v \in A} \ch_{\gG}(v)$, etc. Additionally, we let
$\overline{\an}_{\gG}(A) := \an_{\gG}(A) \cup A$. The subscript $\gG$
is often omitted when it is clear from the context.

\subsection{m-separation}
\label{sec:m-separation}

Introduced by \citet{richardson03_markov_proper_acycl_direc_mixed_graph},
m-separation extends the d-separation criterion
for conditional independence in DAGs \citep{pearl1988book} to
ADMGs. To introduce this concept, we say a \emph{path} in an ADMG $\gG$ is
\emph{(ancestrally) blocked} by $C \subseteq V$ if there exists a
non-collider on the path that is in $C$ or a collider on the path that
is in $C$ or has a descendant in $C$. This is the usual notion of
blocking in the literature \citep{pearl1988book}. Later in the
article, we will introduce a slightly different notion of blocking
that does not depend on the ambient graph and will refer to this
conventional notion as ancestral blocking.

For disjoint sets $A,B,C \subset \sV$, if there exists a path
in $\gG$ from any vertex in $A$ to any vertex in $B$ not (ancestrally)
blocked given $C$, we say
$A$ and $B$ are \emph{m-connected} in $\gG$ given $C$ and denote it by
$A \mconn B \mid C \ingraph{\gG}$; otherwise if all paths from $A$ to
$B$ are (ancestrally) blocked by $C$, we say $A$ and $B$
are \emph{m-separated} by $C$ in $\gG$ and denote it by $\textnot A \mconn B
\mid C\ingraph{\gG}$. While some authors use the independence symbol
`$\indep$' to denote m-separation, our notation more clearly describes
the \emph{type} of the underlying path: (1) the half-arrowheads
indicate that both endpoints are unrestricted in terms of arrowhead
or tail; (2) the wildcard character `$\ast$' means the path can have zero,
one or several colliders. As such, $A \mconn B$ refers to all paths
from $A$ to $B$. The advantages of this notation will become clearer in
\cref{sec:suff-adjustm-sets,sec:m-conn}.

\subsection{Graph marginalization} \label{sec:latent-projection}
In many cases, a graph can contain variables that are not observed (or
observable\footnote{We say a variable is \emph{observed} if the
  variable has been measured and recorded, and we say a variable is
  \emph{observable} if the variable can be measured. Typically, the
  former notion is used when analyzing a study and the latter when
  designing a study. For convenience, we do not differentiate the two
  and use ``observed'' throughout. }) and it is convenient to
contemplate the \emph{marginal graph} that describes the marginal
distribution of the observed variables; this process is often referred
to as \emph{latent projection} in the literature
\citep{pearlTheoryInferredCausation1991}.
Let $\gG = (\sV,\sD,\sB)$ be an ADMG and let $\sV =\tilde{\sV} \cup \sU$ be a
partition (so $\tilde{\sV} \cap \sU = \emptyset$).
The marginalization of $\gG$ onto $\tilde{\sV}$ is the ADMG
$\gG(\tilde{\sV}) = (\tilde{\sV}, \tilde{\sD}, \tilde{\sB})$ with
directed and bidirected edges given by
\begin{itemize}
\item $A \rdedge B$ in $\gG(\tilde{\sV})$ if there exists a directed path in $\gG$ on which all the non-endpoint vertices are in $\sU$ (such paths are denoted as $A
  \overset{\text{via}\,\sU}{\rdpath} B$), and
\item $A \bdedge B$ in $\gG(\tilde{\sV})$ if there exists a path
  from $A$ to $B$ such that the path has two endpoint arrowheads and no
  colliders and has all the non-endpoint vertices in $\sU$ (such paths
  are denoted as $A  \overset{\text{via}\,\sU}{\confarc} B$).
\end{itemize}
Trivially, $A \rdedge B$ in $\gG$ implies $A \rdedge B$ in $\gG(\tilde{\sV})$ and $A \bdedge B$ in $\gG$ implies $A \bdedge B$ in $\gG(\tilde{\sV})$.
Using the squiggly line notation, marginalization can be compactly
defined as
\[
  A
  \begin{Bmatrix}
    \rdedge \\
    \ldedge \\
    \bdedge \\
  \end{Bmatrix}
  B \ingraph{\gG(\tilde{\sV})}
  \quad
  \iff
  \quad
  A
  \begin{Bmatrix}
    \overset{\text{via}\,\sU}{\rdpath} \\
    \overset{\text{via}\,\sU}{\ldpath} \\
    \overset{\text{via}\,\sU}{\fullsquigfull} \\
  \end{Bmatrix}
  B\ingraph{\gG}.
\]
It follows from this definition that marginalization preserves
acyclicity. It is well-known that m-separation is
preserved by marginalization
\citep{verma1990equivalence,evans2018markov}: for disjoint $A, B, C
\subset \tilde{\sV}$, we have
\[
  A \mconn B \mid C \ingraph{\gG(\tilde{\sV})} \quad \iff \quad A
  \mconn B \mid C\ingraph{\gG}.
\]
In \cref{sec:m-conn}, we strengthen this result by showing that
marginalization also preserves some refined forms of m-connections.

A central problem in causal inference is to identify the causal
effect of a treatment variable $X$ on an outcome variable $Y$ under
the causal model represented by an ADMG $\gG$. The most widely used
approach is to adjust for a set of variables $S$ that controls for
confounding. Under a standard causal graphical model (see below), a
set $S \subseteq \sV \setminus \{X,Y\}$ is guaranteed to control for
confounding if it satisfies the \emph{back-door criterion}
\citep{pearlBayesianAnalysisExpert1993} on $\gG$:
\begin{equation} \label{eqs:back-door}
\begin{split}
&\text{1. $S$ contains no descendant of $X$, and}\\
&\text{2. there is no m-connected path from $X$ to $Y$ given $S$ with an arrowhead into $X$.}
\end{split}
\end{equation}
Further, \citet[Theorem 5]{shpitser2010validity} showed that this criterion is
complete in the following sense: given any set $S'
\subseteq \sV \setminus \{X,Y\}$ that controls for confounding under
the causal model represented by $\gG$, the set $S' \setminus \de(X)$
must satisfy the back-door criterion. Because set $S' \setminus \de(X)$
is no bigger than $S'$ and also controls for confounding,
it suffices to only consider those $S$ that do not contain descendants of $X$.
We will stick to this convention for the rest of this paper.

\subsection{Causal models} \label{sec:causal-models}
Heuristically, a directed edge $A \rdedge B$ signifies the presence of
a \emph{direct} causal effect of $A$ on $B$ that is not mediated by
other variables in the graph. A bidirected edge $A \bdedge B$
signifies the presence of
\emph{exogenous correlation} --- a basic form of unmeasured
\emph{confounding} --- between $A$ and $B$ that cannot be explained
away by other variables in the graph; this is closely related to the concept
of endogeneity in the econometrics literature
\parencite{engleExogeneity1983a}. Formally, a causal model can be defined
on ADMGs using the potential outcomes language and unconfoundedness
is often defined as the following conditional independence
\parencite{rosenbaumCentralRolePropensity1983}:
\begin{equation}
  \label{eq:unconfoundedness}
  X \indep Y(x) \mid S \quad \text{for all}~x,
\end{equation}
where $Y(x)$ is the potential outcome of $Y$ when $X$ is set to
$x$. The single-world intervention graph provides an easy way to check
conditional independence involving potential outcomes; in fact, the
back-door criterion is equivalent to m-separation in the single-world
intervention graph \parencite{richardson2013single}. As our main
contribution is graphical, we will refrain from further discussion and
refer the reader to \textcite{shpitser2022multivariate} and
\textcite{zhaoStatisticalCausalModels2025} for recent work on
causal models associated with ADMGs.

\section{Confounding paths, confounding arcs and sufficient adjustment
  sets} \label{sec:suff-adjustm-sets}
Despite the well-established notion of unconfoundedness as the
conditional independence in \eqref{eq:unconfoundedness}, it may come
as a surprise that giving a positive definition of a confounding
variable is more difficult. This difficulty arises partly due to the
lack of monotonicity of unconfoundedness --- a set controls for
confounding but its superset may not.

In this section, we reformulate Pearl's back-door criterion
symmetrically in terms of confounding paths, which are structures in
the graph that induce confounding between the two endpoints. This
allows us to shift the consideration of confounding from vertices to
edges and paths in the graph.

\begin{definition}
  Consider an ADMG $\gG$ over a vertex set $\sV$. For distinct $A,B
  \in \sV$,
  let $\sP[A \mconn B \ingraph{\gG}]$ denote the set of all paths between $A$ and
  $B$ in $\gG$, and $\sP[A \rdpath B \ingraph{\gG}]$ denote the subset
  of all directed paths from $A$ to $B$. Define the set of \emph{confounding paths}
  between $A$ and $B$ in $\gG$ as
  \[
    \sP[A \confpath B \ingraph{\gG}] := \{\pi \in \sP[A \mconn B
    \ingraph{\gG}]:  \text{$\pi$ has two endpoint arrowheads}\}.
  \]
  Those confounding paths without colliders are called \emph{confounding arcs}:
  \[
    \sP[A \confarc B \ingraph{\gG}] := \{\pi \in \mathcal{P}[A
    \confpath B \ingraph{\gG}]: \pi \text{ has no collider} \}.
  \]
\end{definition}

A confounding arc immediately induces \emph{non-causal}
association between the endpoints when none of its non-endpoint
vertices are conditioned on.
A confounding arc between $A$ and $B$ is of the form
\[
  A \ldpath U_1 \rdpath B \quad \text{or} \quad A
  \ldpath U_1 \bdedge U_2 \rdpath B.
\]
As our notation `$\confpath$' suggests, a
confounding path is either a confounding arc or a concatenation of
multiple confounding arcs. In the latter case, the confounding path
may induce non-causal association when all the colliders (or their
descendants) on the path are conditioned on. See also 
\textcite{zhaoMatrixAlgebraGraphical2024} for a systematic treatment
of these graphical concepts using a matrix algebra on the set of walks
in $\gG$.

To further distinguish different types of paths that
are ancestrally unblocked given a set $C$ (see
\Cref{sec:m-separation}), we adopt the following notation.

\begin{definition} \label{def:refined-m}
  For distinct $A,B \in \sV$ and $C \subseteq \sV \setminus \{A,B\}$,
  denote
  \[
    \sP[A \cdots B \mid_a C \ingraph{\gG}] = \{\pi \in \sP_{\gG}[A
                                         \cdots B \ingraph{\gG}]: \pi
                                                 \text{~is ancestrally unblocked
                                                 given $C$}\},
  \]
  where `$\cdots$' can be `$\rdpath$', `$\confarc$', `$\mconn$', `$\confpath$', etc.
\end{definition}

We often omit ``$\textnormal{\bf in} \gG$'' from our
notation when the graph $\gG$ is clear from the
context. The subscript ``$a$'' in $|_a$ is
used to emphasize that we are
considering the conventional notion of (ancestral) blocking, where a
vertex can unblock a path if the path contains a collider that is not the vertex but the vertex's
ancestor. We will consider a slightly different notion of blocking
that does not permit this; see \Cref{def:walk-m-conn} below.

\begin{remark} \label{rem:path}
  Note that conditioning on the empty set changes the meaning of our
  notation when the collection of paths may contain colliders. For a
  given ADMG $\gG$ and $A,B,C$ in \Cref{def:refined-m}, we have
  \begin{equation}
    \label{eq:confarc-is-confpath}
  \sP[A \confarc B \mid_a C] \subseteq \sP[A \confpath B
  \mid_a C].
  \end{equation}
When no variable is conditioned on, an unblocked (confounding) path
must be a (confounding) arc:
\begin{align*}
  \sP[A \confpath B \mid_a \emptyset] = \sP[A \confarc B
  \mid_a \emptyset] =& \sP[A \confarc B] \neq \sP[A \confpath B],\\
  \sP[A \mconn B \mid_a \emptyset] = \sP[A \mconnarc B
  \mid_a \emptyset] =& \sP[A \mconnarc B] \neq \sP[A \mconn B].
\end{align*}
In other words, conditioning on $\emptyset$ means that the unblocked
paths cannot contain colliders.
\end{remark}

Recall that in \Cref{sec:graph-preliminaries} we write the relation
``$A$ and $B$ are m-connected given $C$'' as
\[ A \mconn B \mid C \iff \sP[A \mconn B \mid_a C] \neq \emptyset, \]
and the relation ``$A$ and $B$ are m-separated given $C$'' as
\[\textnot A \mconn B \mid C \iff \sP[A \mconn B \mid_a C] =
  \emptyset. \]
Similarly, we define the following relations using other types of paths.

\begin{definition} \label{def:conf-relation}
  Let $\gG$ be an ADMG over vertex set $\sV$. For any distinct $A,B \in
  \sV$ and $C \subseteq \sV \setminus \{A,B\}$, we write
  \[
    A \begin{Bmatrix} \rdpath \\ \confarc \\
        \confpath \end{Bmatrix} B \mid C \ingraph{\gG}
      \quad \iff \quad
      \sP \left[  A
        \begin{Bmatrix} \rdpath  \\ \confarc \\
          \confpath \end{Bmatrix} B \mid_a C \ingraph{\gG} \right] \neq \emptyset
    \]
    and
    \[ \textnot A \begin{Bmatrix} \rdpath \\ \confarc \\
           \confpath \end{Bmatrix} B \mid C \ingraph{\gG} \quad \iff
         \quad
           \sP \left[  A
  \begin{Bmatrix} \rdpath  \\ \confarc \\
    \confpath \end{Bmatrix} B \mid_a C \ingraph{\gG} \right] =
\emptyset.
\]
\end{definition}

Note that the relations above pertaining to `$\confarc$' and
`$\confpath$' are symmetric in $A$ and $B$; these relations satisfy certain 
graphoid-like properties as discussed in \Cref{sec:graphoid}.
The next result directly follows from definition.

\begin{lemma}[Monotonicity of unblocked arcs] \label{lem:monotone}
  Let $\gG$ be an ADMG over vertex set $\sV$. For $\tilde{C} \subseteq
  C \subseteq \sV \setminus \{A, B\}$, we have
  \begin{align*}
     &A \rdpath B \mid C \ingraph{\gG} \implies A \rdpath B \mid
       \tilde{C} \ingraph{\gG}, \\
    &A \confarc B \mid C \ingraph{\gG} \implies A
  \confarc B \mid \tilde{C} \ingraph{\gG}.
  \end{align*}
\end{lemma}

It is well known that such monotonicity is not true for blocking
confounding paths; a famous counter-example is the M-bias graph
\citep{greenland1999causal}.

Next, we offer our definition of a sufficient adjustment set that is
symmetric in $X$ and $Y$.
\begin{definition}[Adjustment set] \label{def:adjust}
Let $\gG$ be an ADMG over vertex set $\sV$. For distinct $X,Y \in
\sV$, we say $S \subseteq \sV \setminus \{X,Y\}$ is an
\emph{adjustment set} for $X$ and $Y$ if it contains no descendants of $X$ or $Y$,
namely,
\[\textnot X \rdpath S \ingraph{\gG} \quad \text{and} \quad \textnot Y
  \rdpath S \ingraph{\gG}.\]
We say an adjustment set $S$ is \emph{sufficient} if
\[\textnot X \confpath Y \mid S \ingraph{\gG}. \]
Moreover, we say a sufficient adjustment set $S$ is \emph{minimal} if
none of its proper subsets is also sufficient.
\end{definition}

\begin{theorem}[Symmetric back-door criterion] \label{thm:adjust}
Let $\gG$ be an ADMG over vertex set $\sV$ and consider distinct
vertices $X,Y \in \sV$. If $X \rdpath Y \ingraph{\gG}$, then $S$ is a
sufficient adjustment set for $X$ and $Y$ if and only if $S$ satisfies the back-door
criterion in \eqref{eqs:back-door}.
\end{theorem}

\Cref{thm:adjust} shows that
our definition is equivalent to Pearl's back-door criterion given a
very mild assumption that imposes the causal direction.
In view of this, confounder selection boils down to
finding an adjustment set $S$ that blocks all the confounding paths
`$\confpath$' between $X$ and $Y$; thus $S$ necessarily blocks all
confounding arcs `$X \confarc Y$'. However, due to the monotonicity
property (\cref{lem:monotone}), blocking confounding arcs is considerably
easier than blocking confounding paths --- expanding $S$ does
not introduce new unblocked confounding arcs but can introduce new
unblocked confounding paths due to colliders. This observation
motivates our procedure in \Cref{sec:expand} that blocks confounding
paths by iteratively blocking confounding arcs.

\begin{remark} \label{rem:symm}
  Our notion of adjustment set and back-door criterion is symmetric in
  $X$ and $Y$. This is made possible by excluding descendants
  of $X$ or $Y$ in adjustment sets. This requirement is essential in
  the graph expansion algorithm below to avoid inquiring the causal
  direction between variables in the adjustment set. In principle, there may exist 
  a set of covariates that contain certain descendants of $X$ (that do not lie on $X
  \rdpath Y$) that can also control for confounding  --- such a set can be called a ``generalized adjustment set'' (see \citealp{shpitser2010validity,perkovic2015}). Yet, descendants of $X$ are unnecessary to adjust for (rarely seen in practice) and they are also prohibited in Pearl's back-door criterion.
\end{remark}

\section{Properties of graph
  connections/separations} \label{sec:m-conn}

Our approach to confounder selection is based on iterative graph expansion, a
process that inverts a sequence of graph marginalizations. Thus, it is
important to understand how confounding paths behave under
marginalization. In fact, similar to the ordinary
m-connection/separation, the relations in \cref{def:refined-m}
(connections via directed paths, confounding arcs, and confounding
paths) are also preserved by marginalization.

\begin{theorem} \label{thm:preserve}
Let $\gG$ be an ADMG over vertex set $V$. For any  distinct $A, B \in
\sV$, $C \subseteq V \setminus \{A,B\}$, and vertex set $\tilde{V}$ such
that $\{A,B\} \cup C \subseteq \tilde{V} \subseteq V$, we have
\[
  A \begin{Bmatrix} \rdpath \\ \confarc \\
      \confpath \\ \mconn \end{Bmatrix} B \mid C \ingraph{\gG} \quad \iff \quad
    A \begin{Bmatrix} \rdpath \\ \confarc \\
        \confpath \\ \mconn \end{Bmatrix} B \mid C
      \ingraph{\gG(\tilde{V})}.
\]
\end{theorem}

This immediately implies the next result.

\begin{corollary} \label{cor:preserve-two-projections}
Let $\gG, \gG'$ be two marginal graphs of the same ADMG onto different
subsets of vertices. For distinct vertices $A,B$ and vertex set $C$
($C \cap \{A,B\} = \emptyset$) that exist on both graphs, we have
\[
  A \begin{Bmatrix} \rdpath \\ \confarc \\
      \confpath \\ \mconn \end{Bmatrix} B \mid C \ingraph{\gG} \quad \iff \quad
    A \begin{Bmatrix} \rdpath \\ \confarc \\
        \confpath \\ \mconn \end{Bmatrix} B \mid C \ingraph{\gG'}.
\]
\end{corollary}

\begin{remark}
\cref{thm:preserve} can be further strengthened as follows:
connections by arcs (i.e., without colliders) are preserved in
marginalization with matching endpoint arrowheads and tails;
connections by general paths are also preserved by marginalization,
but only \emph{endpoint arrowheads}
are preserved. This means that, for example, the following implication
with only matching endpoint arrowheads is true (this can be seen from our proof
of \Cref{thm:preserve}):
\[
  \sP[A \nosquigfull \ast \fullsquigfull B \mid_a C \ingraph{\gG}]
  \neq \emptyset
  \quad \implies \quad \sP[\halfsquigfull \ast \fullsquigfull B\mid_a C
  \ingraph{\gG(\tilde{\sV})}] \neq \emptyset,
\]
but, in general,
\[
  \sP[A \nosquigfull \ast \fullsquigfull B \mid_a C \ingraph{\gG}]
  \neq \emptyset
  \quad \centernot{\implies} \quad \sP[A \nosquigfull \ast
  \fullsquigfull B \mid_a C  \ingraph{\gG(\tilde{\sV})}] \neq
  \emptyset.
\]
To illustrate the last point, consider
\cref{fig:non-preserve}, where $\gG', \gG''$ are successive marginals
of $\gG$, and $\gF'$ is a marginal graph of $\gF$. We have the following
observations:
\begin{enumerate}[(i)]
\item The statement $$\sP[A \nosquigfull \ast \fullsquigno B \mid_a C]
  \neq \emptyset$$ holds in $\gG$
  (path $A \rdedge E \bdedge F \ldedge B$) and $\gG'$
  (path $A \rdedge C \bdedge F \ldedge B$) but not in
  $\gG''$.

\item The statement $$\sP[A \nosquigfull \ast \fullsquigfull D
  \mid_a C] \neq \emptyset$$
  holds in $\gG$ (path $A \rdedge E \bdedge F
  \ldedge B \rdedge D)$ and $\gG'$ (path $A \rdedge C
  \bdedge F \ldedge B \rdedge D$), but not in
  $\gG''$. Instead, $\sP[A \confpath D \mid_a C] \neq \emptyset$ holds
  in $\gG''$ (path $A \ldedge B \rdedge D$).

\item With wildcard `+' denoting one or more colliders, the statement
  $$\sP[A \confarc + \confarc B \mid_a C] \neq \emptyset$$ holds in
  $\gF$ (path $A \ldedge D \bdedge E \bdedge B$) but not in $\gF'$.
\end{enumerate}
\end{remark}

\begin{figure}[t]
\centering
\begin{tikzpicture}
\tikzset{rv/.style={circle,inner sep=1pt,draw,font=\sffamily},
lv/.style={circle,inner sep=1pt,fill=gray!50,draw,font=\sffamily},
fv/.style={rectangle,inner sep=1.5pt,fill=gray!20,draw,font=\sffamily},
node distance=12mm, >=stealth}
\begin{scope} \node[name=A, rv]{$A$};
\node[name=E, rv, right of=A]{$E$};
\node[name=F, rv, right of=E]{$F$};
\node[name=B, rv, right of=F]{$B$};
\node[name=D, rv, right of=B]{$D$};
\node[name=C, rv, below of=A, xshift=5mm]{$C$};
\draw[->, thick] (A) to (E);
\draw[<->, thick] (E) to (F);
\draw[<-, thick] (F) to (B);
\draw[->, thick] (B) to (D);
\draw[->, thick, bend right] (F) to (A);
\draw[->, thick] (A) to (C);
\draw[->, thick] (E) to (C);
\node[right of=C, xshift=10mm, yshift=-7mm]{$\gG$};
\end{scope} \begin{scope}[xshift=6cm] \node[name=A, rv]{$A$};
\node[name=F, rv, right of=A]{$F$};
\node[name=B, rv, right of=F]{$B$};
\node[name=D, rv, right of=B]{$D$};
\node[name=C, rv, below of=A, xshift=5mm]{$C$};
\draw[->, thick] (F) to (A);
\draw[->, thick] (B) to (F);
\draw[->, thick] (B) to (D);
\draw[->, thick] (A) to (C);
\draw[<->, thick] (F) to (C);
\node[right of=C, xshift=2mm, yshift=-7mm]{$\gG'$};
\end{scope} \begin{scope}[xshift=11cm] \node[name=A, rv]{$A$};
\node[name=B, rv, right of=A]{$B$};
\node[name=D, rv, right of=B]{$D$};
\node[name=C, rv, below of=A, xshift=5mm]{$C$};
\draw[->, thick] (B) to (A);
\draw[->, thick] (B) to (D);
\draw[->, thick] (A) to (C);
\draw[<->, thick, bend right] (A) to (C);
\node[right of=C, xshift=-3mm, yshift=-7mm]{$\gG''$};
\end{scope} \begin{scope}[yshift=-3cm] \node[name=A, rv]{$A$};
\node[name=D, rv, right of=A]{$D$};
\node[name=E, rv, right of=D]{$E$};
\node[name=B, rv, right of=E]{$B$};
\node[name=C, rv, below of=D]{$C$};
\draw[->, thick] (D) to (A);
\draw[<->, thick, bend left] (D) to (E);
\draw[->, thick] (E) to (D);
\draw[<->, thick] (E) to (B);
\draw[->, thick] (D) to (C);
\node[right of=C, xshift=-5mm, yshift=-7mm]{$\gF$};
\end{scope} \begin{scope}[yshift=-3cm, xshift=6cm] \node[name=A, rv]{$A$};
\node[name=D, rv, right of=A]{$D$};
\node[name=B, rv, right of=D]{$B$};
\node[name=C, rv, below of=D]{$C$};
\draw[->, thick] (D) to (A);
\draw[<->, thick] (D) to (B);
\draw[->, thick] (D) to (C);
\node[below of=C, yshift=4mm]{$\gF'$};
\end{scope} \end{tikzpicture}
\caption{$\gG'$ is a marginal of $\gG$ and $\gG''$ is a further
  marginalization of $\gG'$; $\gF'$ is a marginal of $\gF$.}
\label{fig:non-preserve}
\end{figure}

\subsection{Proof of \Cref{thm:preserve}}
\label{sec:proof-crefthm:pr}

We prove \Cref{thm:preserve} by re-characterizing these relations
using unblocked \emph{simple walks} with a slightly different notion
of blocking, denoted as $\sW^s(\cdot \mid \cdot)$. In particular, the
first equivalence in \Cref{thm:preserve} follows from
\Cref{lem:walk-path,lem:walk-preserve} below via the following
equivalence diagram (denote $\tilde{\gG} := \gG(\tilde{V})$):
  \begin{center} \setlength{\tabcolsep}{3pt}
  \begin{tabular}{ccccc}
    $A \rdpath B \mid C\ingraph{\gG}$ & $\overset{\text{Def.}\ \ref{def:conf-relation}}{\iff}$
    & $\sP[A \rdpath B \mid_a C \ingraph{\gG}] \neq \emptyset$
    & $\overset{\text{Lem.}\ \ref{lem:walk-path}}{\iff}$
      & $\sW^s[A \rdpath B \mid C \ingraph{\gG}] \neq
        \emptyset$ \\
                                    & & & & \qquad ${\bigg \Updownarrow}$
                                            \footnotesize Lem.\
                                            \ref{lem:walk-preserve}
    \\
    $A \rdpath B \mid C \ingraph{\,\tilde{\gG}}$ & $\overset{\text{Def.}\ \ref{def:conf-relation}}{\iff}$
    & $\sP[A \rdpath B \mid_a C \ingraph{\,\tilde{\gG}}] \neq \emptyset$
    & $\overset{\text{Lem.}\ \ref{lem:walk-path}}{\iff}$
      & $\sW^s[A \rdpath B \mid C \ingraph\,{\tilde{\gG}}] \neq
        \emptyset$.
  \end{tabular}
\end{center}
Here, $\sW^s[A \rdpath B \mid C \ingraph{\gG}]$ is the collection of
all simple unblocked directed walks from $A$ to $B$ given $C$. Recall
that a walk is a sequence of adjacent edges of any type or
orientation; the notions of (non-ancestral) blocking and simple walks
will be introduced shortly. All other equivalences in
\Cref{thm:preserve} can be proved in exactly the same way.

We briefly explain why reasoning with walks instead of paths
simplifies the proof. Given a path with colliders, one may have to
look beyond the path to determine whether it is ancestrally
blocked. That is, the notion of ancestrally blocked paths is
\emph{relative} to the graph:
the same path may be ancestrally blocked in one graph but not blocked
in another. For example, in
\Cref{fig:path-m-conn-relative}, the path $A \rdedge D
\ldedge B$ is not ancestrally blocked given $C$ in $\gG_1$ but
ancestrally blocked given $C$ in $\gG_2$. This issue can be avoided by
using a slightly different notion of blocking that does not involve
the descendants of any collider. This notion of blocking is used in
the well-known Bayes ball algorithm
\citep{shachterBayesballRationalPastime1998}; see also
\citet{van2019separators}.

\begin{figure}[t]
  \centering
  \begin{subfigure}[b]{0.25\textwidth} \centering
\begin{tikzpicture}
\tikzset{rv/.style={circle,inner sep=1pt,draw,font=\sffamily},
lv/.style={circle,inner sep=1pt,fill=gray!50,draw,font=\sffamily},
fv/.style={rectangle,inner sep=1.5pt,fill=gray!20,draw,font=\sffamily},
node distance=12mm, >=stealth}
\node[name=A, rv]{$A$};
\node[name=D, rv, right of=A]{$D$};
\node[name=B, rv, right of=D]{$B$};
\node[name=C, rv, below of=D]{$C$};
\draw[->, thick] (A) to (D);
\draw[->, thick] (B) to (D);
\draw[->, thick] (D) to (C);
\end{tikzpicture}
\caption{$\gG_1$}
\end{subfigure}
  \begin{subfigure}[b]{0.25\textwidth} \centering
\begin{tikzpicture}
\tikzset{rv/.style={circle,inner sep=1pt,draw,font=\sffamily},
lv/.style={circle,inner sep=1pt,fill=gray!50,draw,font=\sffamily},
fv/.style={rectangle,inner sep=1.5pt,fill=gray!20,draw,font=\sffamily},
node distance=12mm, >=stealth}
\node[name=A, rv]{$A$};
\node[name=D, rv, right of=A]{$D$};
\node[name=B, rv, right of=D]{$B$};
\node[name=C, rv, below of=D]{$C$};
\draw[->, thick] (A) to (D);
\draw[->, thick] (B) to (D);
\end{tikzpicture}
\caption{$\gG_2$}
  \end{subfigure}
\caption{The difference between
  blocking and ancestral blocking.}
\label{fig:path-m-conn-relative}
\end{figure}

\begin{definition}[Blocked walk] \label{def:walk-m-conn}
A walk is said to be \emph{blocked} given a vertex set $C$ if
there exists a non-collider on the walk contained in $C$ or a collider
on the walk not contained in $C$.
\end{definition}

Thus in \Cref{fig:path-m-conn-relative}, the walk $A \rdedge D
\ldedge B$ is blocked given $C$ in both $\gG_1$ and
$\gG_2$. Yet, the walk $A \rdedge D \rdedge C \ldedge
D \ldedge B$ (present in $\gG_1$ but not in $\gG_2$) is not blocked
given $C$, irrespective of the graph that the walk resides in.

Let $\sW[A \mconn B \mid C \ingraph{\gG}]$ denote all unblocked walks
from $A$ to $B$ given $C$ in an ADMG $\gG$. Other types of unblocked
walks are defined in a similar way. Specifically, we write
\begin{align*}
  \sW[A \confpath B \mid C] &:= \{\pi \in \sW[A \mconn B
                                    \mid C]:  \text{$\pi$
  has two endpoint arrowheads}\}, \\
  \sW[A \confarc B \mid C] &:= \{\pi \in \sW[A
  \confpath B \mid C]: \pi \text{ has no collider} \}.
\end{align*}
Note that $\sP[A \confarc B \mid_a C] \subseteq \sW[A
\confarc B \mid C]$ is true because confounding arcs have no
colliders, but
$\sP[A \confpath B \mid_a C] \not \subseteq \sW[A \confpath B \mid C]$ due
to the different notions of blocking being applied. Nevertheless,
it is shown below that the
  existence of ancestrally unblocked paths, as one may
  expect, is equivalent to the existence of unblocked walks.

Since we are interested in the preservation of endpoint arrowheads by
marginalization, we will focus on \emph{simple walks}; see
\cref{rem:simple-walk} below.

\begin{definition} \label{def:walk-simple}
A walk is \emph{simple} if its endpoints are visited only
once by the walk.
\end{definition}

An anonymous referee pointed out that the proof below would still go through if
it is further required that a simple walk contains no repeated edge in
the same direction (in which case the walk must have finite length).

We add a superscript `$s$' to $\sW[\cdots]$ to indicate that only simple walks are considered.
For example, $\sW^s[A \confpath B \mid C \ingraph{\gG}]$ denotes all
unblocked confounding simple walks in $\gG$ from $A$ to $B$ given $C$,
i.e.,
\[
  \sW^s[A \confpath B \mid C \ingraph{\gG}] := \left\{\pi \in \sW[A \confpath
  B \mid C \ingraph{\gG}]:\text{$\pi$ is simple}\right\}.
\]
It is easy to see that there exists an unblocked walk between
$A$ and $B$ given $C$ if and only if there exists an unblocked
simple walk between $A$ and $B$ given $C$, that is,
\[
    \sW_{\gG}(A \mconn B \mid C) = \emptyset
    \quad \iff \quad \sW_{\gG}^s(A \mconn B \mid C) = \emptyset.
  \]

As mentioned earlier, we prove \cref{thm:preserve} by applying the
next two key lemmas.
\cref{lem:walk-path} shows that the existence of certain types of
paths that are not ancestrally blocked is equivalent to that of
unblocked simple walks. \cref{lem:walk-preserve} shows that certain
types of unblocked simple walks are preserved by marginalization. The
proofs of these lemmas are deferred to \Cref{apx:proofs}.

\begin{lemma} \label{lem:walk-path}
  Let $\gG$ be an ADMG over vertex set $\sV$. For any distinct $A,
  B \in \sV$ and $C \subseteq \sV \setminus \{A,B\}$, we have
  \[
  \sP \left[  A
  \begin{Bmatrix} \rdpath  \\ \confarc \\
    \confpath \\ \mconn \end{Bmatrix} B \mid_a C \ingraph{\gG} \right] \neq \emptyset
\quad \iff \quad
  \sW^s \left[ A
  \begin{Bmatrix} \rdpath  \\ \confarc \\
    \confpath \\ \mconn \end{Bmatrix} B \mid C \ingraph{\gG} \right]
\neq \emptyset.
  \]
\end{lemma}

\begin{lemma} \label{lem:walk-preserve}
Let $\gG$ be an ADMG over vertex set $V$. For any  distinct $A, B \in
\sV$, $C \subseteq V \setminus \{A,B\}$ and vertex set $\tilde{V}$ such
that $\{A,B\} \cup C \subseteq \tilde{V} \subseteq V$, we have, \[
  \sW^s \left[  A \begin{Bmatrix} \rdpath \\ \confarc \\
      \confpath \\ \mconn \end{Bmatrix} B \mid C \ingraph{\gG} \right]
  \neq \emptyset ~~
  \iff ~~
  \sW^s \left[ A \begin{Bmatrix} \rdpath \\ \confarc \\
      \confpath \\ \mconn \end{Bmatrix} B \mid C \ingraph{\gG(\tilde{V})} \right] \neq \emptyset.
\]
\end{lemma}

\begin{remark} \label{rem:simple-walk}
For an unblocked walk that is not simple, there may not exist a
corresponding ancestrally unblocked path with matching endpoint
arrowheads. Consider the
graph $C \rdedge A \rdedge B$. Although the walk $A \ldedge C
\rdedge A \rdedge B$ is not blocked given the empty set, there is no
ancestrally unblocked path from $A$ to $B$ starting with `$A
\ldedge$'.
\end{remark}

\subsection{District criterion}
\label{sec:district-criterion}

From \Cref{thm:preserve}, we derive the following simple but useful
criterion for confounder selection. In what follows, we write `$A
\bdedge \ast \bdedge B$' if there exists a path from $A$ to $B$ consisting of bidirected edges only. Similarly, we write `$A \colliderconn B$' if there exists a path from
$A$ to $B$ on which every non-endpoint vertex is a collider.

\begin{corollary} \label{cor:preserve}
For an ADMG $\gG$, consider two distinct vertices $A, B$ and a vertex set $C$ such that $A, B \notin C$.
We have
\[
  A \begin{Bmatrix} \rdpath \\ \confarc \\
       \confpath \\ \mconn \end{Bmatrix} B \mid C \ingraph{\gG} \quad \iff
     \quad A \begin{Bmatrix} \rdedge \\ \bdedge \\
               \samedist \\ \colliderconn \end{Bmatrix} B \ingraph{\gG(\{A,B\} \cup C)}.
\]
\end{corollary}
\begin{proof}
Apply \cref{thm:preserve} with $\tilde{V} = \{A,B\} \cup C$.  Because
$C$ contains all vertices other than $A,B$ in $\gG(\tilde{V})$, an
unblocked arc from $A$ to $B$ given $C$ must be a single directed
or bidirected edge. Further, any non-endpoint vertex on any
unblocked path from $A$ to $B$ given $C$ in $\gG(\tilde{V})$
must be a collider.
\end{proof}

We make two remarks on this result. First, recall that
in the literature of graphical models, a \emph{district} in an
ADMG is defined as a maximal set of vertices connected by bidirected edges (see, e.g., \citealp{evans2014markovian}). That
is, $A$ and $B$ are in the same district if and only if $A \samedist B$.
By the third equivalence relation in \Cref{cor:preserve}, an
adjustment set $S$ is sufficient for $A$ and $B$ if and only if $A$
and $B$ are in different districts of $\gG(\{A,B\}\cup S)$,
the smallest marginal graph. We
call this the \emph{district criterion} for confounder selection. It
provides a simple and useful way to check whether a set of selected
confounders is indeed sufficient. Second, the last equivalence
relation in \Cref{cor:preserve} says that $A$ and $B$ are m-connected
given $C$ if and only if $A$ and $B$ are ``collider-connected'' in the
smallest marginal graph. This concept naturally arises in considering
the inverse covariance matrix for linear structural equation models;
see
\textcite{kosterValidityMarkovInterpretation1999} and \textcite{zhaoMatrixAlgebraGraphical2024}.

\section{Iterative graph expansion} \label{sec:expand}
\label{sec:graph-expansion}

\subsection{Primary adjustment set}

A primary adjustment set is any adjustment set that blocks all the
\emph{confounding arcs} `$\confarc$' between two vertices.
They are the building blocks when we try to find sufficient adjustment sets
by iteratively expanding the graph, a procedure that is introduced in the
next subsection.

\begin{definition}[Primary adjustment set]
Let $\gG$ be an ADMG with vertex set $V$. Given two distinct vertices
$A,B \in V$ and a set $S \subseteq V \setminus \{A,B\}$,
an adjustment set $C$ for $A,B$ is called \emph{primary
  relative to $S$} if
\[\textnot A \confarc B \mid S \cup C \ingraph{\gG}.\] When $S =
\emptyset$, we simply say
that $C$ is a primary adjustment set for $A$ and $B$. Further, $C$ is called
\emph{minimal primary} (relative to $S$) if none of its proper subsets is
primary (relative to $S$).\end{definition}

By definition, if $C$ is a (minimal) primary adjustment set for $A$
and $B$, it is also a (minimal) primary adjustment set for $B$ and
$A$. Further, because any ancestrally unblocked confounding arc is an
ancestrally unblocked path (see
\eqref{eq:confarc-is-confpath}), any sufficient adjustment set must be
primary but the reverse may not be true. For
example, in the graph $A \bdedge D \bdedge B$, the adjustment set
$\{D\}$ for $A$ and $B$ is primary but not sufficient.

\begin{lemma}
  Given an ADMG $\gG$, if $C$ is a minimal primary adjustment set for
  $A$ and $B$ relative to $S$, then $C \cap S = \emptyset$ and $C
  \subseteq \an_{\gG}(A) \cup \an_{\gG}(B)$.
\end{lemma}
\begin{proof}
  Suppose $\textnot A \confarc B \mid S \cup \tilde{C}$ for some
  $\tilde{C}$. Let $C = \tilde{C} \cap (\an(A) \cup \an(B)) \setminus
  S$. It is easy to show that any confounding arc $A \confarc B$ that
  is not blocked given $S \cup \tilde{C}$ remains not blocked given $S
  \cup C$. That is, $\textnot A \confarc B
  \mid S \cup C$. The result then follows from the definition of minimal
  primary adjustment set.
\end{proof}

\begin{theorem} \label{thm:completeness}
Let $\gG$ be an ADMG over vertex set $V$. Suppose $S \subset V$ is a
minimal sufficient adjustment set for distinct $X, Y \in V$. Then, for any
$\tilde{S} \subset S$, there exist distinct $Z_1,Z_2 \in \tilde{S} \cup
\{X,Y\}$ such that
\begin{equation} \label{eqs:completeness}
 Z_1 \bdedge Z_2 \ingraph{\gG(\tilde{S} \cup \{X,Y\})} \quad
 \text{and} \quad \textnot Z_1 \bdedge Z_2 \ingraph{\gG(S \cup
 \{X,Y\})}.
\end{equation}
  Moreover, there exists a non-empty minimal primary
  adjustment set $C \subseteq S \setminus \tilde{S}$ for $Z_1$ and $Z_2$
  relative to $\tilde{S} \setminus \{Z_1, Z_2\}$.
\end{theorem}
Before we prove \Cref{thm:completeness}, note that by
\Cref{cor:preserve}, \cref{eqs:completeness} is equivalent
to
\[
  Z_1 \confarc Z_2 \mid \tilde{S} \cup \{X,Y\} \setminus \{Z_1,Z_2\}
  \ingraph{\gG} \quad \text{and} \quad \textnot Z_1 \confarc Z_2 \mid S
  \cup \{X,Y\} \setminus \{Z_1,Z_2\} \ingraph{\gG}.
\]
Because $Z_1, Z_2 \in \tilde{S} \cup \{X,Y\}$ and $\tilde{S}$ is an
adjustment set, we know $Z_1$ and $Z_2$ are not descendants of $X$ or
$Y$. Therefore, we can safely remove $\{X,Y\}$ from the conditioning
sets and \eqref{eqs:completeness} is also equivalent to
\begin{equation} \label{eqs:completeness-2}
 Z_1 \confarc Z_2 \mid \tilde{S} \setminus \{Z_1,Z_2\}
  \ingraph{\gG} \quad \text{and} \quad \textnot Z_1 \confarc Z_2 \mid S \setminus
  \{Z_1,Z_2\} \ingraph{\gG}.
\end{equation}
To interpret the conclusions of this theorem, imagine trying to find a
minimal sufficient
adjustment set $S$ by iteratively adding vertices. Suppose $\tilde{S}
\subseteq S$ is our current estimate. As long as $\tilde{S}$ is not
yet sufficient, we can find $Z_1, Z_2 \in \tilde{S} \cup \{X,Y\}$ such
that there is a confounding arc between them that is not blocked by
$\tilde{S}$. To block it, we can try to find a primary adjustment set
$C$ for $Z_1, Z_2$ relative to $\tilde{S}$ and add $C$ to $\tilde{S}$ and
make $\tilde{S} \cup C$ our next estimate. This theorem states that
for at least one choice of $C$, the new estimate $\tilde{S} \cup C$
will remain a subset of (and possibly the same as) $S$. If
$\tilde{S} \cup C$ is still not sufficient, we may iterate the process
until we eventually obtain the set $S$. In other words, \Cref{thm:completeness}
essentially shows that all minimal sufficient adjustment
sets for $(X,Y)$ can be found by recursively expanding a bidirected
edge with its minimal primary adjustment sets; this is the basis of
\Cref{alg:confound-select} below.

\begin{proof}[Proof of \Cref{thm:completeness}]
  Because $\tilde{S}$ is a proper subset of
  a minimal sufficient adjustment set $S$, $\tilde{S}$ is an
  adjustment set but is not sufficient. By \cref{cor:preserve}, we have
  \[ X \confpath Y \mid \tilde{S} \ingraph{\gG} \quad \implies \quad X
    \samedist Y \ingraph{\gG(\tilde{S} \cup \{X,Y\})}. \]
  Therefore, there exists a bidirected path
  \[ X=D_0 \bdedge \dots \bdedge D_{k} = Y \ingraph{\gG(\tilde{S} \cup
      \{X,Y\})}, \quad k \geq 1. \]
  We claim that there exists $j \in \{0,\dots,k-1\}$ such that $D_j
  \centernot{\bdedge} D_{j+1}$ in $\gG(S \cup
  \{X,Y\})$. Otherwise, the same bidirected path also appears in
  $\gG(S \cup \{X,Y\})$ (because $\tilde{S} \subset S$), so by
  \cref{cor:preserve},
\[ X \samedist Y \ingraph{\gG(S \cup \{X,Y\})} \quad \implies
  \quad X \confpath Y \mid S \ingraph{\gG}, \]
contradicting the assumption that $S$ is sufficient. By our choice, we have
$D_j, D_{j+1} \in \tilde{S} \cup \{X,Y\}$ and $D_j \bdedge
D_{j+1}$ in $\gG(\tilde{S} \cup \{X,Y\})$. Hence, we have shown the
existence of $(D_j, D_{j+1})$, rewritten as $(Z_1, Z_2)$, as desired.

For the second conclusion, observe that \cref{eqs:completeness-2} holds, 
that is, $S \setminus \{Z_1,Z_2\}$ is
primary for $(Z_1, Z_2)$ but
$\tilde{S} \setminus \{Z_1,Z_2\}$ is not. The existence of such
a minimal primary adjustment set $C$ then follows from the definition.
\end{proof}

\subsection{Iterative graph expansion algorithm}
\label{sec:graph-expans-algor}

We now introduce our procedure \textsc{ConfounderSelect} that
is based on two sub-routines: \textsc{SelectEdge} and \textsc{FindPrimary}. See
\Cref{alg:confound-select} for the pseudo-code of an implementation
using a priority queue $\mathcal{Q}$ that admits a \textsc{Pop} and a \textsc{Push}
operation; see also \cref{sec:pseudo-code-recurs} for a recursive
version of the algorithm that is shorter but less flexible. We leave
the choices of the priority index and of the \textsc{SelectEdge}
subroutine to the next subsection, as they do not affect the validity of
the algorithm. In practice, it is through the \textsc{FindPrimary}
subroutine that the information about the underlying graph is elicited
from the user; we discuss its implementation in \cref{sec:find-primary}.

In this algorithm, all the possible bidirected edges $\bar{S} \times \bar{S}$ in the
current graph $\gG(\bar{S})$, where $\bar{S}:=S \cup \{X,Y\}$ and $S$ is our working
adjustment set, are partitioned into three groups:
\begin{enumerate}
\item the set $\mathcal{B}_n$ that contains all the bidirected edges
  that are absent or already blocked (that is, $(A,B) \in
  \mathcal{B}_n$ means that we are already certain that $\textnot A
  \bdedge B \mid \bar{S}$ in the current graph and thus in all its
  expansions by \Cref{lem:monotone});
\item the set $\mathcal{B}_y$ that contains all the bidirected edges that
  exist or assumed to exist (i.e., confounding that cannot be controlled for, which we draw as $\color{blue} \barrow$), which the algorithm will not attempt to
  block; and
\item the set $\mathcal{B}_u := (\bar{S} \times \bar{S}) \setminus
  (\mathcal{B}_n \cup \mathcal{B}_y)$ that contains all the uncertain
  bidirected edges (drawn as $\dbarrow$)that the algorithm will attempt to block.
\end{enumerate}
In other words, the algorithm maintains $\mathcal{B}_u \cup
\mathcal{B}_y$ as a superset of the bidirected edges in
$\gG(\bar{S})$. In each iteration, the subroutine $\Call{SelectEdge}{X,Y,S,\mathcal{B}_y,\mathcal{B}_n}$ selects an
uncertain bidirected edge $\pi$ from $\mathcal{B}_u$.
The subroutine $\textsc{FindPrimary}$ is then called to find the
primary adjustment sets for the two end points of $\pi$ relative to the current
adjustment set.
If $\pi$ is already blocked by the current adjustment set (i.e., the empty set
is returned as a primary adjustment set), then $\pi$ is moved to $\mathcal{B}_n$.
Otherwise, the algorithm
attempts to move $\pi$ to $\mathcal{B}_n$ by expanding $\pi$ with
every primary adjustment set that is returned by the subroutine
$\textsc{FindPrimary}$. Finally, it attempts to make no expansion by moving
$\pi$ to $\mathcal{B}_y$. Because the vertex set $\sV$ is assumed
to be finite, this algorithm will eventually terminate.

\begin{theorem}[Soundness and completeness of iterative graph expansion] \label{thm:alg}
Let $\gG$ be an ADMG over vertex set $\sV$. The following two statements hold
for any distinct $X,Y \in \sV$.
  \begin{enumerate}
  \item  Suppose for any $A,B \in \sV$ and $S' \subseteq \sV \setminus \{A,B\}$,
  when $\Call{FindPrimary}{A,B; S'} \neq \emptyset$,
  every $C \in \Call{FindPrimary}{A,B; S'}$ is
    a primary adjustment set for $A$ and $B$ relative to $S'$ in $\gG$. Then
    every set in the output of $\Call{ConfounderSelect}{X,Y}$ is a
    sufficient adjustment set for $X$ and $Y$.
  \item Suppose further that $\Call{FindPrimary}{A,B; S'}$ contains every minimal primary adjustment set for $A$ and $B$ relative to $S'$ in $\gG$. Then the output of $\Call{ConfounderSelect}{X,Y}$ contains all the minimal sufficient adjustment sets for $X$ and $Y$.
  \end{enumerate}
\end{theorem}
\begin{proof}
  Statement 1 says that the graph expansion algorithm is sound. It follows
  from the district criterion (\Cref{cor:preserve}), as an adjustment
  set $S$ is only added to the output $\mathcal{R}$ in lines 9-11 of
  \Cref{alg:confound-select} when $X$ and $Y$ are not connected by
  bidirected edges in $\mathcal{B}_y \cup \mathcal{B}_u$ (and hence in
  $\gG(\bar{S})$). Statement 2
  says that the graph expansion algorithm is complete for identifying
  minimal sufficient adjustment sets. It directly follows from
  \Cref{thm:completeness}.
\end{proof}

\begin{algorithm}[t]
  \caption{Confounder selection via iterative graph expansion}
  \label{alg:confound-select}
  \begin{algorithmic}[1] \Procedure{ConfounderSelect}{$X$, $Y$}
      \State $\mathcal{R} = \{ \}$ \Comment{Set of sufficient
        adjustment sets}
      \State $\mathcal{Q} = \Call{PriorityQueue}{(\emptyset, \emptyset,\emptyset)}$
      \Comment{Initial graph has a possible edge $X \dbarrow Y$}
      \While{$\mathcal{Q} \neq \emptyset$}
        \State $(S,\mathcal{B}_y,\mathcal{B}_n) = $
        \Call{Pop}{$\mathcal{Q}$}  \Comment{\parbox[t]{.34\linewidth}{$S$: current adjustment set\\ $\mathcal{B}_y$: edges not to be expanded \\$\mathcal{B}_n$: absent or blocked edges}}
        \State $\bar{S} = S \cup \{X,Y\}$
        \If{$X \samedist Y$ by edges in $\mathcal{B}_y$} \Comment{Fails the district criterion}
            \State \textbf{continue}
        \ElsIf{$\textnot X \samedist Y$ by edges in $(\bar{S}
            \times \bar{S}) \setminus \mathcal{B}_n$} \Comment{Satisfies the district criterion}
          \State $\mathcal{R} = \mathcal{R} \cup \{S\}$
          \State \textbf{continue}
        \EndIf
        \State $(A,B) = \pi = $ \Call{SelectEdge}{$X$, $Y$, $S$, $\mathcal{B}_y$, $\mathcal{B}_n$} \Comment{$\pi$ is selected from $(\bar{S} \times \bar{S}) \setminus (\mathcal{B}_y \cup \mathcal{B}_n)$}
        \State $\mathcal{L} = $ \Call{FindPrimary}{$(A,B); S \setminus
          \{A,B\}$}
        \If{$\emptyset \in \mathcal{L}$}
                \State \Call{Push}{$\mathcal{Q}$, ($S$, $\mathcal{B}_y$,
          $\mathcal{B}_{n} \cup \{\pi\}$)} \Comment{$\pi$ need not be expanded}
        \Else
        \For{$C \in \mathcal{L}$}
          \State \Call{Push}{$\mathcal{Q}$, ($S \cup C$, $\mathcal{B}_y$,
            $\mathcal{B}_{n} \cup \{\pi\}$)} \Comment{Expand $\pi$ by each primary adjustment set}
        \EndFor
        \State \Call{Push}{$\mathcal{Q}$, ($S$, $\mathcal{B}_y \cup \{\pi\}$,
          $\mathcal{B}_{n}$)} \Comment{Not to expand $\pi$}
       \EndIf
      \EndWhile
      \State \Return $\mathcal{R}$
    \EndProcedure
    \end{algorithmic}
  \end{algorithm}

\subsection{Practical considerations} \label{sec:choice-graph-expans}
\Cref{alg:confound-select} provides a sound and complete template
for confounder selection, which relies on specifying subroutines \textsc{FindPrimary} and \textsc{SelectEdge}
as well as a priority index specified for the queue $\mathcal{Q}$.
 In practice, instead of finding all the
(minimal) sufficient adjustment sets, often the goal is to find \emph{one} such
set \emph{quickly}, i.e., with only a few attempts of graph expansion. Further, when there are multiple edges to choose from, we can pick one judiciously to improve the efficiency for estimating the causal effect with adjustment. In what follows, we discuss how to design the priority index and \textsc{SelectEdge} to optimize these aforementioned \emph{secondary objectives} of confounder selection. We leave the
implementation of \textsc{FindPrimary} to the next subsection.

To find a (minimal) sufficient adjustment set quickly, we recommend the following \emph{min-cut} strategy. Suppose the priority queue $\mathcal{Q}$ is implemented such that (1) \textsc{Pop}($\mathcal{Q}$) returns an element with the lowest index, and (2) in case of a tie, returns the element that is last pushed to $\mathcal{Q}$. We choose the priority index to be
\begin{multline*}
\text{min-cut}(S,\mathcal{B}_y,\mathcal{B}_n) := \text{minimal number of edges removed from }\\ (\bar{S} \times \bar{S}) \setminus (\mathcal{B}_y \cup \mathcal{B}_n) \text{ to disconnect $X$ and $Y$}.
\end{multline*}
In case of $X \samedist Y$ by edges in $\mathcal{B}_y$, let $
\text{min-cut}(S,\mathcal{B}_y,\mathcal{B}_n) := \infty$. Because
$\text{min-cut}=0$ whenever the district criterion is satisfied (line
9, \cref{alg:confound-select}), this choice prioritizes those
candidates needing the fewest number of expansions. Accordingly, we
recommend a subroutine \textsc{SelectEdge}($X$, $Y$, $S$, $\mathcal{B}_y$,
$\mathcal{B}_n$) that returns an edge that lies \emph{on} the min-cut.
When there is more than one edge that lies on the min-cut, we further recommend
choosing one that is closest to $Y$ to improve the efficiency of the corresponding adjustment
estimator for the effect of $X$ on $Y$, with the intuition being that a covariate
close to $Y$ can help predict $Y$ to reduce variance; see
\citet{henckel2022} and \citet{rotnitzky2020efficient} for further
discussion.
This strategy is
adopted by the examples in \Cref{sec:an-example} and \Cref{sec:unobserved},
where the edge chosen from the popped graph is marked in red.

\subsection{Finding primary adjustment sets} \label{sec:find-primary}

\begin{algorithm}[t]
  \caption{Finding primary adjustment sets for $A$ and $B$ relative to $S$}
  \label{alg:primary}
  \begin{algorithmic}[1] \Procedure{FindPrimary}{$A$, $B$; $S$, \texttt{minimal\_only=False}}
      \State $\mathcal{L} = \{\}$ \Comment{Set of primary adjustment sets relative to $S$}
      \State $\mathcal{T} = \Call{PriorityQueue}{S}$
      \Comment{Initialize with the current adjustment set}
      \While{$\mathcal{T} \neq \emptyset$}
        \State $T = \Call{Pop}{\mathcal{T}}$
                \State $C = \Call{CommonCause}{A,B; T}$ \Comment{Find
                  a common cause that is not blocked by $T$}
                \If {$C ~\text{is}~ \textsc{null}$} \Comment{$T$ is primary}
                        \State $L = T \setminus S$
                        \If {\texttt{minimal\_only}} \Comment{Remove redundant vertices}
                                \For {$Z \in L$}
                                        \If {$\Call{CommonCause}{A,B;
                                              S \cup L \setminus \{Z\}}$ is \textsc{null}}
                                                \State $L = L \setminus \{Z\}$
                                        \EndIf
                                \EndFor
                        \EndIf
                        \State $\mathcal{L} = \mathcal{L} \cup \{L\}$
                      \Else
                        \If {$\Call{IsObserved}{C}$}
                                \State $\Call{Push}{\mathcal{T}, T \cup \{C\}}$
                        \Comment{Add the common cause to $T$}
                        \Else
                        \For{$M_i \in \Call{FindMediator}{A,B;C,S}$}
                            \State $\Call{Push}{\mathcal{T}, T \cup
                              M_i}$
                            \Comment{Add mediators to $T$}
                        \EndFor
                        \EndIf
                \EndIf
      \EndWhile
      \State \Return $\mathcal{L}$
    \EndProcedure
    \end{algorithmic}
\end{algorithm}

Given two vertices $A,B$ and a set $S$ such that $A,B
\notin S$, recall that an adjustment set $S'$ for $A,B$ is primary relative to $S$
if $\textnot A \confarc B \mid S \cup S'$. Suppose $\bar{\gG}$ is the underlying causal
ADMG over vertex set $\bar{\sV}$. By \cref{thm:preserve}, $\textnot A \confarc B \mid S \cup S'$
holds if and only if the set of common causes of $A$ and $B$ not blocked by $S \cup S'$ is empty, that is
\[ \mathcal{C}(A,B; S \cup S')  := \{C \in \bar{\sV}: C \rdpath A \mid
  S, S' \text{ and } C \rdpath B \mid S, S' \ingraph{\bar{\gG}} \} = \emptyset.\]

Consider \cref{alg:primary} for an iterative procedure that identifies
primary adjustment sets. It uses three subroutines: \textsc{CommonCause},
\textsc{FindMediator}, and \textsc{IsObserved}.
If $\mathcal{C}(A,B;S) \neq \emptyset$, the subroutine
$\Call{CommonCause}{A,B; S}$ returns a vertex from
$\mathcal{C}(A,B;S)$, i.e., a common cause of $A$ and $B$ whose causal
paths to $A$ and $B$ are not blocked by $S$; otherwise,
$\Call{CommonCause}{A,B; S}$ returns \textsc{null}.
For $C \in \mathcal{C}(A,B;S)$, the subroutine
$\Call{FindMediator}{A,B;C,S}$ then returns a collection of
\emph{observed} adjustment sets $\{M_i \subseteq \sV \}$ of $A,B$ such
that \[\textnot C \rdpath A \mid S, M_i \ingraph{\bar{\gG}} \quad \text{or} \quad \textnot C \rdpath B
  \mid S, M_i \ingraph{\bar{\gG}}.\]
These subroutines can be formulated in terms of questions in plain language, which we discuss in \cref{sec:discuss}.

The procedure \textsc{FindPrimary} calls \textsc{CommonCause}
iteratively in a similar way that \textsc{ConfounderSelect} calls
\textsc{FindPrimary}. If \textsc{CommonCause} is \textsc{null}, then the
current adjustment set is primary and one can further trim this set to
obtain a minimal primary adjustment set when the option
\texttt{minimal\_only} is set to \texttt{True} (see line 9-17 in
\Cref{alg:primary}). Otherwise, we need to
block the confounding arc $A
\leftsquigarrow C \rdpath B$ by adding
either $C$ or mediators on the $C \rdpath A$ or $C
\rdpath B$ paths (obtained by \textsc{FindMediator}) to
$S$, if these variables can be observed.
When implementing this algorithm, it may be desirable to use a
priority queue to store candidate sets to prioritize smaller
adjustment sets.

\section{Example}
\label{sec:an-example}

\begin{figure}[t]
\centering
\begin{tikzpicture}
\tikzset{rv/.style={circle,inner sep=1pt,draw,font=\sffamily},
lv/.style={circle,inner sep=1pt,fill=gray!50,draw,font=\sffamily},
fv/.style={rectangle,inner sep=1.5pt,fill=gray!20,draw,font=\sffamily},
node distance=12mm, >=stealth}
\node[rv] (X) {$X$};
\node[right of=X] (l0) {};
\node[rv, right of=l0] (I) {$I$};
\node[right of=I] (l1) {};
\node[rv, right of=l1] (Y) {$Y$};
\node[rv, above of=X, yshift=5mm] (T) {$T$};
\node[rv, above of=I] (C) {$C$};
\node[rv, above of=l0] (P) {$P$};
\node[rv, above of=l1] (N) {$N$};
\node[rv, above of=P] (G) {$G$};
\node[rv, above of=G] (O) {$O$};
\node[rv, above of=C, yshift=5mm] (F) {$F$};
\node[rv, above of=N, yshift=12mm] (E) {$E$};
\node[rv, above of=Y] (W) {$W$};
\node[rv, above of=W] (D) {$D$};
\draw[->, thick] (X) -- (I);
\draw[->, thick] (I) -- (Y);
\draw[->, thick] (T) -- (X);
\draw[->, thick] (T) -- (P);
\draw[->, thick] (C) -- (P);
\draw[->, thick] (C) -- (I);
\draw[->, thick] (N) -- (I);
\draw[->, thick] (N) -- (Y);
\draw[->, thick] (O) -- (T);
\draw[->, thick] (O) -- (F);
\draw[->, thick] (E) -- (F);
\draw[->, thick] (E) -- (D);
\draw[->, thick] (E) -- (N);
\draw[->, thick] (D) -- (W);
\draw[->, thick] (D) -- (N);
\draw[->, thick] (W) -- (Y);
\draw[->, thick] (F) -- (N);
\draw[->, thick] (F) -- (G);
\draw[->, thick] (G) -- (X);
\end{tikzpicture}
\caption{A causal DAG in \citet{shrier2008reducing}.}
\label{fig:warmup}
\end{figure}

We use the DAG in \Cref{fig:warmup} taken from
\citet{shrier2008reducing} to illustrate the iterative graph expansion
algorithm. In this example, the treatment $X$ is warm-up prior to
exercise and the outcome $Y$ is injury, and this DAG was used by
\citet{shrier2008reducing} to select confounders using the back-door
criterion. In \cref{sec:unobserved}, we provide another example that involves unobserved variables.

When applying the graph expansion procedure, each iteration is indexed
by a tuple $(S,
\mathcal{B}_y, \mathcal{B}_n)$ associated with a graph that consists
of bidirected edges on vertices $\bar{S} = S \cup \{X,Y\}$. Hereafter, we use
$\color{blue} \barrow$ to indicate bidirected edges in $\mathcal{B}_y$
that are not to be expanded, and $\dbarrow$ to indicate bidirected
edges in $(\bar{S} \times \bar{S}) \setminus (\mathcal{B}_y \cup
\mathcal{B}_n)$ that can be further expanded. In addition, we
highlight the edge $\pi$ chosen by \textsc{SelectEdge} with
$\color{red} \dbarrow$ and indicate the
graph popped from the queue with a box. The superscript of each box
indicates the min-cut of the popped graph and in each iteration a graph with the smallest min-cut is popped from the queue; see
\Cref{sec:choice-graph-expans}. For this example, we assume
the subroutine \textsc{FindPrimary} returns all minimal primary
adjustment sets (i.e., option \texttt{minimal\_only=True}), so by \cref{thm:alg} the procedure is guaranteed to find all the
minimal sufficient adjustment sets.
\begin{enumerate}
\item We start with
\[  \mathcal{Q} = \Biggl[\; \boxed{
\;\Biggr], \quad \mathcal{R}=\{\}.
\end{multline*}

\item In the boxed graph, we have $\textnot E \confarc Y \mid F$ (see \cref{fig:warmup}) and
  hence \[\emptyset \in \textsc{FindPrimary}((E,Y); \{F\},
    \texttt{True}).\] We get
\begin{multline*} \small
\mathcal{Q} = \Biggl[\; \boxed{%
\;\Biggr], \quad \mathcal{R}=\{\}.
\end{multline*}

\item Because $X$ and $Y$ become disconnected in the boxed graph above, a sufficient
  adjustment set $\{E,F\}$ is identified. The algorithm proceeds with
\begin{multline*} \small
\mathcal{Q} = \Biggl[\; \boxed{%
\;\Biggr], \quad \mathcal{R}=\{{\color{blue} \{E,F\}}\}.
\end{multline*}
\item For the boxed graph above, we have $\textsc{FindPrimary}(X,F; \emptyset, \texttt{True}) = \left\{\{T\},\{O\} \right\}$ and hence
\begin{multline*} \small
\mathcal{Q} = \Biggl[\; \boxed{%
\;\Biggr], \quad \mathcal{R}=\{ \{E,F\}\}.
\end{multline*}

\item Noting that $\textnot T \confarc X \mid F$ in \cref{fig:warmup}, we immediately have
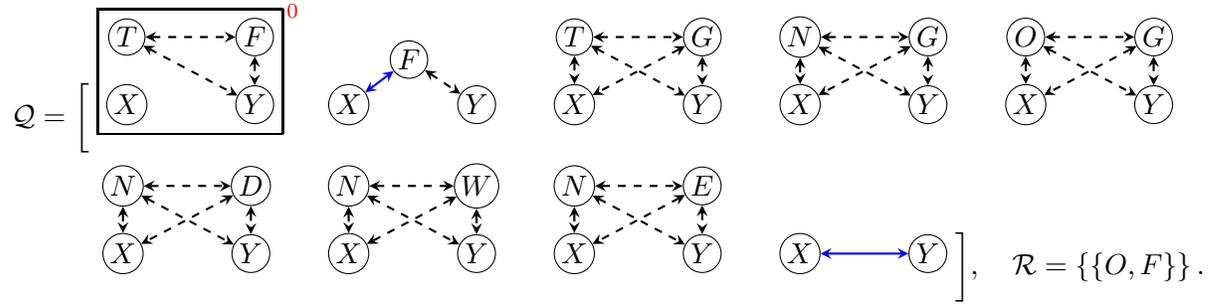
\begin{multline*} \small
\mathcal{Q} = \Biggl[\; \boxed{%
\;\Biggr], \quad \mathcal{R}=\{\{E,F\}\}.
\end{multline*}

\item Given that $\textnot X \samedist Y$ in the boxed graph above, we find
  another sufficient adjustment set $\{T,F\}$. We now have
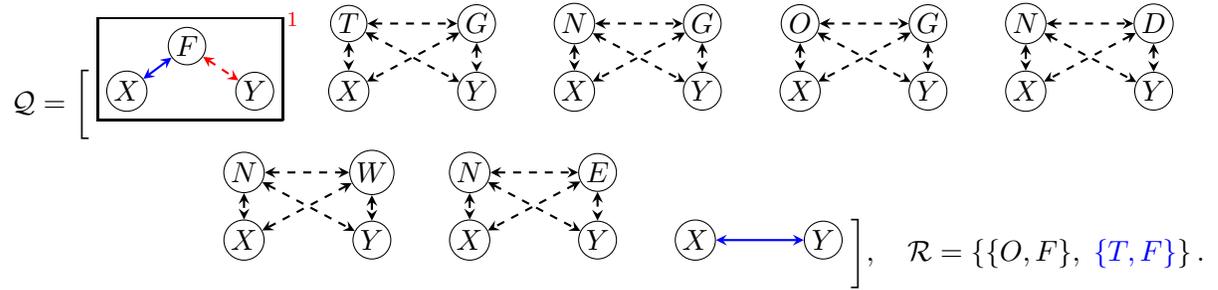
\begin{multline*} \small
\mathcal{Q} = \Biggl[\; \boxed{%
\;\Biggr], \quad \mathcal{R}=\left\{\{E,F\}, {\color{blue} \{T,F\}}\right\}.
\end{multline*}
Besides $\{E,F\}$ and $\{T,F\}$, the algorithm will proceed to find all the other minimal sufficient adjustment sets: $\{O,G\}$, $\{O,F\}$, $\{D,N\}$, $\{N,W\}$ and $\{G,T\}$. For brevity, we omit the remaining steps, which follow analogously.
\end{enumerate}

\section{Discussion} \label{sec:discuss}
We have shown that confounder selection, reformulated as blocking
confounding paths $X \confpath Y$, can be reduced to iteratively
finding primary adjustment sets that block confounding arcs $A
\confarc B$, where $(A,B)$ is either $(X,Y)$ or any other pair of variables
introduced in this process. Finding primary adjustment sets can
be further reduced to finding common causes and finding mediators.
The graph expansion procedure presented as \cref{alg:confound-select} provides
a systematic approach to confounder selection. In practice, the user or domain expert
interacts with the procedure by providing information to the subroutine in \cref{alg:primary}.
At least in principle, the information can be elicited by asking the following questions for various choices of variables $A,B$ and set $S$.
\begin{enumerate}[(i)]
\item Is there a common cause $C$ of $A$ and $B$ such that neither its effect on $A$ nor its effect on $B$ is fully mediated by $S$?

\item If the answer is ``yes'', is $C$ observed?

\item If the answer to the first question is ``yes'' but to the second question is ``no'', then what are some of the observed variables that, when combined with those variables already in $S$, fully
mediate the causal effect of $C$ on $A$ or the causal effect of $C$ on
$B$?
\end{enumerate}
Clearly, answering these questions does not require one's familiarity
with the theory of causal graphical models beyond the basic concepts
of causes and mediators.

We have focused on finding the sufficient adjustment set for a single
treatment variable $X$ and a single outcome variable
$Y$. \citet{shpitser2010validity} and \citet{perkovic2015} provided
generalized adjustment criteria for multiple $X$ and $Y$ variables
when the graph is given. It would be useful to extend the
iterative graph expansion algorithm to this more general problem.

We conclude our paper with several remarks. First, it is evident that
this graph expansion procedure only requires partial knowledge about
the underlying graph. In particular, the procedure never inquires
about the directed edges between the variables in the working
graph. Second, by construction
our method guards against the so-called collider bias or M-bias
\citep{greenland1999causal}. More broadly, the district criterion can
be employed to identify such bias for any given adjustment
set. Finally, we would like to stress again that confounder selection
is crucial for both the design and the analysis of an
observational study. As pointed out by \citet{rubin2008objective},
``for objective causal
inference, design trumps analysis''. Building on the success of causal graphs to
represent structural assumptions, more work on
using causal graphs to aid the design of observational studies is much needed.

\bibliographystyle{imsart-nameyear}

\appendix
\crefalias{section}{appendix}
\clearpage{}

\section{Recursive version of the graph expansion algorithm}
\label{sec:pseudo-code-recurs}

The recursive version is shorter but less flexible; see \Cref{alg:confound-select-recursion} below.

\begin{algorithm}[h]
  \caption{Confounder selection via graph expansion (recursive version)}
  \label{alg:confound-select-recursion}
  \begin{algorithmic}[1] \Procedure{ConfounderSelect}{$X$, $Y$}
      \State $\mathcal{R} = \emptyset$
      \Procedure{GraphExpand}{$S$, $\mathcal{B}_y$, $\mathcal{B}_n$}
        \State $\bar{S} = S \cup X \cup Y$
        \If{$X \samedist Y$ by edges in
            $\mathcal{B}_y$}
          \State \Return
        \ElsIf{$\textnot X \samedist Y$ by edges in $(\bar{S}
            \times \bar{S}) \setminus \mathcal{B}_n$}
          \State $\mathcal{R} = \mathcal{R} \cup \{S\}$
          \State \Return
        \EndIf
        \State $(A,B) = \pi = $ \Call{SelectEdge}{$X$, $Y$, $S$, $\mathcal{B}_y$, $\mathcal{B}_n$}
        \State $\mathcal{L} = $ \Call{FindPrimary}{$(A,B); S \setminus \{A,B\}$}
\For{$C$ in \Call{FindPrimary}{$\pi$, $S \setminus \{A, B\}$}}
          \Comment{Primary adjustment sets for $\pi$ given $S$}
          \State \Call{GraphExpand}{$S \cup C$, $\mathcal{B}_y$,
            $\mathcal{B}_{n} \cup \{\pi\}$}
        \EndFor
        \State \Call{GraphExpand}{$S$, $\mathcal{B}_y \cup \{\pi\}$,
          $\mathcal{B}_{n}$}
\EndProcedure
      \State \Call{GraphExpand}{$\emptyset$, $\emptyset$,
        $\emptyset$}
      \Comment{Initial graph has a possible edge $X \bdedge Y$}
      \State \Return $\mathcal{R}$
    \EndProcedure
    \end{algorithmic}
  \end{algorithm}

\section{Proofs} \label{apx:proofs}
\subsection{Proof of \cref{thm:adjust}} \label{sec:proof-adjust}
\begin{proof}
Pearl's back-door criterion \eqref{eqs:back-door} can be stated as (i) $S \cap
\de(X) = \emptyset$ and (ii) $\sP[X \fullsquigfull \ast \fullsquighalf Y \mid_a S] = \emptyset$. Given $X \rdpath Y \ingraph{\gG}$,
it follows that $\text{(i)} \iff S \cap (\de(X) \cup \de(Y)) =
\emptyset$. Because $\sP[X \confpath Y \mid_a S] \subseteq \sP[X \fullsquigfull \ast \fullsquighalf Y \mid_a S]$, we have $\text{(ii)} \implies \textnot X \confpath Y \mid
S$ (see \cref{def:conf-relation}).
It remains to show that the reverse is also true. To prove by
contradiction, suppose we have
\[\sP[X \confpath Y \mid_a S] = \emptyset, \quad \sP[X \fullsquigfull \ast \fullsquighalf Y \mid_a S] \neq \emptyset,\]
which then implies $\sP[X \leftrightsquigarrow \ast \fullsquigno Y \mid_a S] \neq \emptyset$. By $X \rdpath Y$ and acyclicity of the graph, any such path must have one or more colliders. Let $C$ be the collider on the path that is closest to $Y$. Then, the m-connection implies $C \in \bar{\an}(S)$ and further $Y \in \an(S)$, contradicting (i).
\end{proof}

\subsection{Proof of \Cref{lem:walk-path}} \label{sec:proof-crefl-path}
\begin{proof}
  {\bf `$\Rightarrow$' direction:} By definition, $\sP[A \rdpath B \mid_a C]
  \subseteq \sW^s[A \rdpath B \mid C]$ and $\sP[A \confarc B \mid_a C]
  \subseteq \sW^s[A \confarc B \mid C]$ because there are no
  colliders. We are left to show the last implication, i.e.\ the
  existence of ancestrally unblocked confounding paths implies the
  existence of unblocked confounding simple walks.

  Let $\tau \in \sP[A \confpath B \mid_a C]$ and suppose, without loss
  of generality, $\tau$
  contains a sequence of colliders $\tilde{C}_1, \dots, \tilde{C}_k$ with
  $\tilde{C}_1$ closest to $A$ ($k \geq 1$). In other words, $\tau$
  looks like
  \[
    A \confarc \tilde{C}_1 \confarc \dotsb \confarc \tilde{C}_k
    \confarc B.
  \]
  We associate with each $\tilde{C}_i$ a vertex $C_i$. Because
  $\tau$ is m-connected, we have $\tilde{C}_i \in \overline{\an}(C)$ for $i = 1,\dots,k$.
  If $\tilde{C}_i \in C$, let $C_i = \tilde{C}_i$; otherwise, choose a vertex
  $C_i \in C$ that is among the closest to $\tilde{C}_i$ by directed
  edges and expand $\tau$ at $C_i$ as:
  \begin{equation}
    \label{eq:added-segment}
    \dotsb \confarc \tilde{C}_i \rightsquigarrow C_i \leftsquigarrow
    \tilde{C}_i \confarc \dotsb.
  \end{equation}
  By choosing the segments $\tilde{C}_i \rightsquigarrow C_i$ to be a
  shortest directed path from $\tilde{C}_i$ to $C_i$ (and $C_i
  \leftsquigarrow \tilde{C}_i$ to be its reverse), this introduces no
  new vertices in $C$ other than $C_i$. Repeating this expansion for
  every $i = 1,\dotsc,k$, we obtain a confounding walk, denoted as
  $\pi_0$, that looks like
  \[
    A \confarc C_1 \confarc \dotsb \confarc C_k
    \confarc B.
  \]
  This walk is not blocked given $C$, i.e.\ $\pi_0 \in \sW[A
  \confpath B \mid C]$, because all of its colliders are in
  $C$ and none of its non-colliders is in $C$.

  However, $\pi_0$ is
  not necessarily a simple walk because $A$ or $B$ may appear in an added
  segment in \eqref{eq:added-segment}. Suppose $A$ appears in the
  segment $\tilde{C}_i \rightsquigarrow C_i \leftsquigarrow
  \tilde{C}_i$; clearly $A \neq C_i$ because $C_i
  \in C$ but $A \not \in C$. Thus, the walk $\pi_0$ looks like
  \[
    A \confarc C_1 \confarc \dotsb \confarc \tilde{C}_i \rightsquigarrow A \rightsquigarrow C_i
    \leftsquigarrow A \leftsquigarrow \tilde{C}_i \confarc \dotsb
    \confarc C_k \confarc B.
  \]
  We may then ``restart'' the walk by considering the subwalk
  starting from the last $A$ in the last display. Similarly, if $B$
  appears in this segment and $\pi_0$ looks like
  \[
    A \confarc C_1 \confarc \dotsb \confarc \tilde{C}_i \rightsquigarrow B \rightsquigarrow C_i
    \leftsquigarrow B \leftsquigarrow \tilde{C}_i \confarc \dotsb
    \confarc C_k \confarc B,
  \]
  we may consider the sub-walk ending at the first $B$. By repeating
  this operation, we arrive at a simple, m-connected walk $\pi$ with
  both arrowheads at $A$ and $B$ preserved. In other words, $\pi \in
  \sW^s[A \confpath B \mid C]$.

  It is easy to see that, by the same expansion and deduplication
  operations, we can construct an unblocked simple walk $\pi \in
  \sW^s[A \mconn B \mid C]$ from an ancestrally unblocked path $\tau \in
  \sP[A \mconn B \mid_a C]$.

  \medskip {\bf `$\Leftarrow$' direction:} Consider $\pi \in \sW^s[A \rdpath B \mid
  C]$. Because $\gG$ is acyclic, $\pi$
  must contain no repeated vertices so $\pi \in \sP[A \rdpath B \mid_a
  C]$ ($\pi$ contains no collider). Consider $\pi \in \sW^s[A \confarc
  B \mid C]$, so $\pi$ looks like
  \[
    A \ldpath U_1 \rdpath B \quad \text{or} \quad A
    \ldpath U_1 \bdedge U_2 \rdpath B.
  \]
  If $\pi$
  contains repeated vertices, let $D$ be the repeated vertex closest
  to $A$ in $\pi$. Because $\gG$ is acyclic, $D$ must appear exactly
  once in the segment $A \ldpath U_1$ and exactly once in the
  segment $U_1
  \rdpath B$ or $U_2 \rightsquigarrow B$. In other words, $\pi$
  looks like
  \begin{align*}
    &A \ldpath D \ldpath U_1 \rdpath D \rdpath B,
    \quad A \ldpath D \bdedge U_2 \rdpath D
    \rdpath B, \\
    &A \ldpath D \ldpath U_1
    \bdedge D \rdpath B, \quad  \text{or} \quad A
    \ldpath D \ldpath U_1 \bdedge U_2 \rdpath D
    \rdpath B.
  \end{align*}
  In either case, we can obtain an unblocked walk $\tau$ that looks
  like $A \ldpath D \rdpath B$ by concatenating the first and last
  segments of the walk. Because $D$ is the closest repeated vertex to
  $A$ in $\pi$, this new walk has no repeated vertices, i.e.\ $\tau \in
  \sP[A \confarc B \mid C]$.

  Now consider $\pi \in \sW^s[A \confpath B \mid C]$. Suppose a
  vertex $D$ appears more than once in $\pi$. Because $\pi$ is simple,
  we know $D \not \in \{A, B\}$. Motivated by the above
  argument for confounding arcs, consider the following deduplication
  operation. For any simple walk $\lambda$ from $A$ to $B$ with a repeated
  vertex $D$, let $\lambda'$ be the concatenation of the sub-walk from
  $A$ to the first occurrence of $D$ and the subwalk from the last
  occurrence of $D$ to $B$. To visualize this, $\lambda$ can always be
  represented as
  \[
    A \halfsquigfull \ast \fullsquighalf D \mconn D \halfsquigfull \ast
    \fullsquighalf B.
  \]
  The deduplicated walk $\lambda'$ then looks like
  \[
    A \halfsquigfull \ast \fullsquighalf D \halfsquigfull \ast
    \fullsquighalf B.
  \]

  Suppose that $\lambda$ satisfies the following two properties:
  \begin{align*}
  \text{(P1) }& \text{all occurrences of the same vertex are of the
                same type, i.e., a collider or a non-collider},\\
  \text{(P2) }& \text{all non-colliders are not in $C$ and all
              colliders are contained in $\overline{\an}(C)$}.
  \end{align*}
  Then, we claim that the map $\lambda \mapsto \lambda'$ preserves these properties.
  To show this, consider the following two cases.
  \begin{enumerate}[(i)]
  \item $D$ is a collider on $\lambda$, so $\lambda$ looks like
  \[
    A \halfsquigfull \ast \fullsquigfull D \confpath D \fullsquigfull
    \ast
    \fullsquighalf B.
  \]
    It is obvious that $D$ is also a collider on
    $\lambda'$. Because $\lambda$ satisfies (P1), we have $D \in
    \overline{\an}(C)$. It is easy to see that $\lambda'$ satisfies
    (P1) and (P2).
  \item $D$ is not a collider on $\lambda$. If $D$ is also a
    non-collider on $\lambda'$, clearly (P1) and (P2) hold for
    $\lambda'$. If $D$ becomes a new collider on $\lambda'$, then
    $\lambda$ looks like
    \[
      A \halfsquigfull \ast \fullsquigfull D \nosquigfull \ast
      \fullsquigno D \fullsquigfull \ast \fullsquighalf B.
    \]
    Observe that on $\lambda$ there exists one or more colliders
    between the first and last occurrence of $D$, and $\lambda$
    contains a directed path from $D$ to a collider. This implies that
    $D \in \overline{\an}(C)$. Thus, (P1) and (P2) continue to hold for
    $\lambda'$.
  \end{enumerate}

  We now return to the simple, confounding walk $\pi \in \sW^s[A
  \confpath B \mid C]$. Let $\pi'$ denote the image
  of $\pi$ under the deduplicating operation above. Because $\pi$ is
  not blocked given $C$, it is obvious that $\pi$ satisfies
  (P1) and (P2). Thus, $\pi'$ also satisfies (P1) and (P2). By iteratively
  applying the deduplication, we eventually eliminate all repeated
  vertices on $\pi$ and arrive at a path $\tau$ from $A$ to $B$ that
  satisfies (P1) and (P2). It immediately follows from (P2) that $\tau$ is
  not blocked given $C$. Moreover, both endpoint arrowheads of $\pi$
  are preserved by the deduplication above. Thus, $\tau \in \sP[A
  \confpath B \mid_a C]$.

  It is easy to see that we can use the same deduplication process to
  construct an ancestrally unblocked path $\tau \in
  \sP[A \mconn B \mid_a C]$ from an unblocked simple walk $\pi \in
  \sW^s[A \mconn B \mid C]$. This completes our proof.
\end{proof}

\subsection{Proof of \Cref{lem:walk-preserve}}
\label{sec:proof-crefl-pres}
\begin{proof}
  We first consider the preservation from $\gG$ to $\tilde{\gG}$.
  Let $\pi \in \sW^s[A \mconn B \mid C \ingraph{\gG}]$. Its
  marginalization onto $\tilde{V}$ is a walk $\pi'$
  in $\tilde{\gG}$ defined as follows. Let $D_0 = A, D_1, \dots, D_k,
  D_{k+1} = B$ ($k \geq 0$) be all vertices on $\pi$ that are
  contained in $\tilde{V}$. By denoting $U = V \setminus \tilde{V}$, $\pi$ looks like
  \begin{equation}
    \label{eq:pi-before-proj}
    A = D_0 \overset{\text{via}\,U}{\mconnarc} D_1
    \overset{\text{via}\,U}{\mconnarc} \dotsb
    \overset{\text{via}\,U}{\mconnarc} D_k
    \overset{\text{via}\,U}{\mconnarc} D_{k+1} = B,
  \end{equation}
  where between every $(D_i, D_{i+1})$ there is an unblocked arc
  because $U \cap C = \emptyset$. The marginalization of each arc
  $D_i
  \overset{\text{via}\,U}{\mconnarc} D_{i+1}$ on $\tilde{V} = V \setminus U$
  is defined as
  \[
    D_i
    \begin{Bmatrix}
      \overset{\text{via}\,U}{\rdpath} \\
      \overset{\text{via}\,U}{\ldpath} \\
      \overset{\text{via}\,U}{\fullsquigfull} \\
    \end{Bmatrix}
    D_{i+1}
    \quad
    \mapsto
    \quad
    D_i
    \begin{Bmatrix}
      \rdedge \\
      \ldedge \\
      \bdedge \\
    \end{Bmatrix}
    D_{i+1},
  \]
  where the corresponding edge on the RHS must exist in $\tilde{\gG}$
  by the definition of marginal graph.
  Thus, the marginalization of $\pi$ is a simple walk $\pi'$ in
  $\tilde{\gG}$ that looks like
  \begin{equation}
    \label{eq:pi-after-proj}
    A = D_0 \halfstraighalf D_1 \halfstraighalf \dotsb \halfstraighalf
    D_k \halfstraighalf D_{k+1} = B.
  \end{equation}
  Thus, we conclude that $\pi' \in \sW^s[A \mconn B \mid C
  \ingraph{\,\tilde{\gG}}]$.
  Further, because the marginalization as defined above preserves endpoint
  arrowheads, $D_i$ is a collider on $\pi'$ if and only if $D_i$ is a
  collider on $\pi$. Consequently, we also have
  \[
  \pi \in \sW^s \left[ A \begin{Bmatrix} \rdpath \\ \confarc \\
      \confpath \end{Bmatrix} B \mid C \ingraph{\gG}  \right] \quad
  \implies \quad
  \pi' \in \sW^s \left[ A \begin{Bmatrix} \rdpath \\ \confarc \\
      \confpath \end{Bmatrix} B \mid C \ingraph{\,\tilde{\gG}}
  \right].
\]

  Now we prove the reverse direction.
  Given any $\pi' \in \sW^s[A \mconn B \mid C
  \ingraph{\,\tilde{\gG}}]$, we want to find a
  simple, m-connected walk $\pi$ in $\gG$. Walk $\pi'$ can be written as
  \eqref{eq:pi-after-proj} for some $A = D_0, D_1,\dots,D_k, D_{k+1}=B
  \in \tilde{V}$. By inverting the marginalization, we obtain a walk $\pi$ in
  $\gG$ that looks like \eqref{eq:pi-before-proj}. It is easy to see
  that $\pi$ is simple because $A,B \not \in U$ and $A, B \not
  \in \{D_1,\dots,D_k\}$ since the $\pi'$ is simple. Further, because
  each segment between $D_i$ and $D_{i+1}$ is an arc and contains no
  collider, $\pi$ and $\pi'$ have the same set of colliders. Therefore,
  we conclude that $\pi \in \sW^s[A \mconn B \mid C
  \ingraph{\gG}]$. Further, by
  the fact that $\pi$ and $\pi$ have the same endpoint arrowheads/tails and the
  same set of colliders, we also have
  \[ \pi' \in \sW^s \left[  A \begin{Bmatrix} \rdpath \\ \confarc \\
      \confpath \end{Bmatrix} B \mid C \ingraph{\,\tilde{\gG}}  \right] \quad
  \implies \quad
  \pi \in \sW^s \left[ A \begin{Bmatrix} \rdpath \\ \confarc \\
                           \confpath \end{Bmatrix} B \mid C
                         \ingraph{\gG}  \right].\]
  This completes our proof.
\end{proof}

\section{Graphoid-like properties} \label{sec:graphoid}

We consider extending the ternary relations in
\cref{def:conf-relation} to disjoint sets $A, B, C \subset \sV$. For
example, we write
\[ A \confpath B \mid C \ingraph{\gG} \quad \iff \quad a \confpath b
  \mid C \ingraph{\gG} \text{ for some $a \in A$, $b \in B$} \]
and
\[ \textnot A \confpath B \mid C \ingraph{\gG} \quad \iff \quad \textnot a
  \confpath b \mid C \ingraph{\gG} \text{ for every $a \in A$, $b \in B$}. \]
In addition, following the convention in the literature, we assume
$\textnot A \confpath \emptyset \mid C$ holds for every disjoint $A$ and
$C$. Extension for `$\confarc$' is made
similarly. To avoid clutter, we use comma in place of `$\cup$' in
these relations, e.g., $A \confpath B,C \mid D$ stands for $A
\confpath B \cup C \mid D$.

\begin{definition}[Graphoid] \label{def:graphoid}
Let $\mathcal{J}$ be a collection of tuples $\langle A, B \mid C \rangle$ for $A,B,C$ that are disjoint subsets of a ground set $\sV$. For disjoint $A,B,C,D \subset \sV$, consider the following properties:
\begin{enumerate}[(i)]
\item triviality: $\langle A, \emptyset \mid C \rangle$ for every disjoint $A,C \subset \sV$;
\item symmetry: $\langle A, B \mid C \rangle \implies \langle B, A \mid C \rangle$;
\item decomposition: $\langle A, B \cup C \mid D \rangle \implies \langle A, B \mid D \rangle \text{ and } \langle A, C \mid D \rangle$;
\item weak union: $\langle A, B \cup C \mid D \rangle \implies \langle A, B \mid C \cup D \rangle$;
\item contraction: $\langle A, C \mid D \rangle \text{ and } \langle A, B \mid C \cup D\rangle \implies \langle A, B \cup C \mid D \rangle$;
\item intersection: $\langle A, B \mid C \cup D \rangle \text{ and } \langle A, C \mid B \cup D \rangle \implies \langle A, B \cup C \mid D \rangle$;
\item composition: $\langle A, B \mid D \rangle \text{ and } \langle A, C \mid D \rangle \implies \langle A, B \cup C \mid D \rangle$.
\end{enumerate}
We say $\mathcal{J}$ is a \emph{semi-graphoid} over $\sV$, if it
satisfies (i)--(v); further, we say $\mathcal{J}$ is a \emph{graphoid}
over $\sV$, if it satisfies (i)--(vi), and finally, a
\emph{compositional graphoid} over $\sV$, if it satisfies (i)--(vii).
\end{definition}

For any ADMG, it is known that m-separation ($\textnot \cdot \mconn
\cdot \mid \cdot$) is a compositional graphoid
\citep{sadeghi2014markov}. However, separation relations in terms of 
`$\confarc$' and `$\confpath$' are weaker than semi-graphoids, 
as stated in the next result.

\begin{theorem}\label{thm:graphoid}
Relations $\mathcal{J}_{\confarc}:=\{\langle A, B \mid C\rangle:
\textnot A
\confarc B \mid C \ingraph{\gG} \}$ and $\mathcal{J}_{\confpath}:=\{
\langle A, B \mid C \rangle: \textnot A \confpath B \mid C \ingraph{\gG} \}$
satisfy properties (i)--(iv) and (vii) in \Cref{def:graphoid} for any
ADMG $\gG$. However, there exist ADMGs in which (v) and (vi) fail for
both relations.
\end{theorem}
\begin{proof}
Properties (i)--(iii) and (vii) directly follow from definition.

We claim that $\mathcal{J}_{\confarc}$ satisfies (iv) weak union. By
(iii), $\textnot A \confarc B, C \mid D$ implies $\textnot A \confarc
B \mid D$, which leads to $\textnot A \confarc B \mid C, D$ by \cref{lem:monotone}.

We claim that $\mathcal{J}_{\confpath}$  also satisfies (iv) weak union. We show its contraposition $A \confpath B \mid C, D \implies A \confpath B, C \mid D$. From $A \confpath B \mid C, D$, we have $a \confarc b \mid C, D$ for some $a \in A$ and some $b \in B$ in $\gG$. Let $\tilde{\gG} := \gG(\{a,b\} \cup C \cup D)$. By \cref{cor:preserve}, we have $a \samedist b \mid C, D$ in $\tilde{\gG}$. Let $\pi: a=v_0 \bdedge \dots \bdedge v_{k+1} = b$ with $k \geq 0$ be any $\samedist$ path between $a$ and $b$ in $\tilde{\gG}$ that is m-connected given $C, D$. If $k=0$, we have $a \bdedge b \mid C, D$ in $\tilde{\gG}$, which by \cref{cor:preserve} implies $a \confarc b \mid C, D$ in $\gG$. It further follows from \cref{lem:monotone} that $a \confarc b \mid D$ in $\gG$ and hence $A \confpath B \mid D$ in $\gG$.
Now we suppose $k \geq 1$. By definition of $\tilde{\gG}$, we have $v_1, \dots, v_k \in C \cup D$. If $v_1, \dots, v_k \in D$, then we have $a \samedist b \mid D$ and hence $a \confpath b \mid D$ in $\tilde{\gG}$. By \cref{thm:preserve}, it follows that $a \confpath b \mid D$ in $\gG$ and hence $A \confpath B \mid D$ in $\gG$.
Otherwise, let $1 \leq k' \leq k$ be the first vertex such that $v_{k'} \in C$. Observing that $a \bdedge \dots \bdedge v_{k'}$ is m-connected given $D$, we have $a \samedist v_{k'} \mid D$ in $\tilde{\gG}$ and hence $a \confpath v_{k'} \mid D$ in $\tilde{\gG}$. By \cref{thm:preserve}, it follows that $A \confpath C \mid D$ in $\gG$.

We show that $\mathcal{J}_{\confpath}$ does not satisfy (v)
contraction or (vi) intersection on some ADMG. Consider the simple
example $A \ldedge C \rdedge B$ and let $D = \emptyset$. Given
$\textnot A \confpath C \mid \emptyset$ and $\textnot A \confpath B
\mid C$, (v)
contraction would imply $\textnot A \confpath B,C \mid \emptyset$, which
fails here. Also, given $\textnot A \confpath B \mid C$ and $\textnot
A \confpath C \mid B$, (vi) intersection would imply $\textnot A
\confpath B,C \mid \emptyset$, which again fails. Similarly, because
$\mathcal{J}_{\confpath}$ implies $\mathcal{J}_{\confarc}$, (v)
contraction or (vi) intersection for $\mathcal{J}_{\confarc}$ would
imply $\textnot A \confarc B,C \mid \emptyset$, which fails. Hence, relation
$\mathcal{J}_{\confarc}$ also fails to satisfy (v) and (vi).
\end{proof}

\section{An example with unobserved variables} \label{sec:unobserved}
We illustrate \cref{alg:confound-select} with another DAG in
\Cref{fig:unobserved}, where $U_1$ and $U_2$ are
unobserved. For this case, $\{B,D\}$ is the only minimal adjustment
set. We use $\color{blue} \barrow$ to draw those edges in
$\mathcal{B}_y$ that are no longer expanded and use $\dbarrow$ to draw
those edges in $(\bar{S} \times \bar{S}) \setminus (\mathcal{B}_y \cup
\mathcal{B}_n)$ that can be further expanded. The min-cut strategy in
\cref{sec:choice-graph-expans} is adopted. The graph popped from the
queue is drawn with a box with its min-cut index marked in the upper
right corner. To improve the efficiency of adjustment, subroutine \textsc{SelectEdge} chooses an edge on a min-cut 
that is closest to $Y$; the edge $\pi$ chosen by \textsc{SelectEdge} is
highlighted with $\color{red} \dbarrow$. Suppose subroutine
\textsc{FindPrimary} returns all the minimal primary adjustment sets
(i.e., with option \texttt{minimal\_only=True}),
so that the algorithm is guaranteed to find $\{B,D\}$ by
\cref{thm:alg}.

\begin{figure}[!htb]
\centering

\caption{Variables $U_1, U_2$ are unobserved.}
\label{fig:unobserved}
\end{figure}

\begin{enumerate}
\item We start with
\begin{align*}
\mathcal{Q} = \Biggl[\; \boxed{%
\;\Biggr], \\
\mathcal{R} &=\{\}.
\end{align*}

\item Note that $\textsc{FindPrimary}((C,X); \{B\}, \texttt{True}) = \left\{\{D\} \right\}$. This generates two cases: (1) $C {\color{red} \dbarrow} X$ is expanded (and removed) by introducing $D$, and (2) $C {\color{red} \dbarrow} X$ is not expanded and turned into $C {\color{blue} \barrow} X$. With one graph popped and two pushed, we now have
\begin{align*}
\mathcal{Q} &= \Biggl[\; \boxed{%
\;\Biggr], \\
\mathcal{R}&=\{\}.
\end{align*}

\item Further, we have $\textnot B \confarc X \mid C, D$ and hence
\begin{align*}
\mathcal{Q} &= \Biggl[\; \boxed{%
\;\Biggr], \\
\mathcal{R}&=\{\}.
\end{align*}

\item From the boxed graph above, we identify $\{B,C,D\}$ as a sufficient adjustment set. We move to the next graph with the smallest min-cut in 
\begin{align*}
\mathcal{Q} &= \Biggl[\; \boxed{%
\;\Biggr], \\
\mathcal{R} &= \{{\color{blue} \{B,C,D \}}\}.
\end{align*}

\item Because $\textsc{FindPrimary}((D,Y); \{B\}, \texttt{True}) = \left\{\{C\}\right\}$, we get two graphs: (1) expanding $D {\color{red} \dbarrow} Y$ by introducing $C$ and (2) turning $D {\color{red} \dbarrow} Y$ to $D {\color{blue} \barrow} Y$. The above becomes
\begin{align*}
\mathcal{Q} &= \Biggl[\; \boxed{%
\;\Biggr],\\
\mathcal{R} &=\{\{B,C,D \}\}.
\end{align*}

\item Due to $C \ldedge U_2 \rdedge Y$, we turn $C {\color{red} \dbarrow} Y$ into $C {\color{blue} \barrow} Y$ and note that the selected graph's min-cut becomes 3. We have
\begin{align*}
\mathcal{Q} &= \Biggl[\; \boxed{%
\;\Biggr],\\
\mathcal{R} &=\{\{B,C,D \}\}.
\end{align*}

\item Observing that $\textnot D \confarc X \mid B$, the above becomes
\begin{align*}
\mathcal{Q} &= \Biggl[\; \boxed{%
\;\Biggr],\\
\mathcal{R} &=\{\{B,C,D \}\}.
\end{align*}

\item Further, due to $\textnot B \confarc X \mid D$, we have
\begin{align*}
\mathcal{Q} &= \Biggl[\; \boxed{%
\;\Biggr],\\
\mathcal{R} &=\{\{B,C,D \}\}.
\end{align*}

\item From the graph selected above, the minimal sufficient adjustment set $\{B,D\}$ is found:
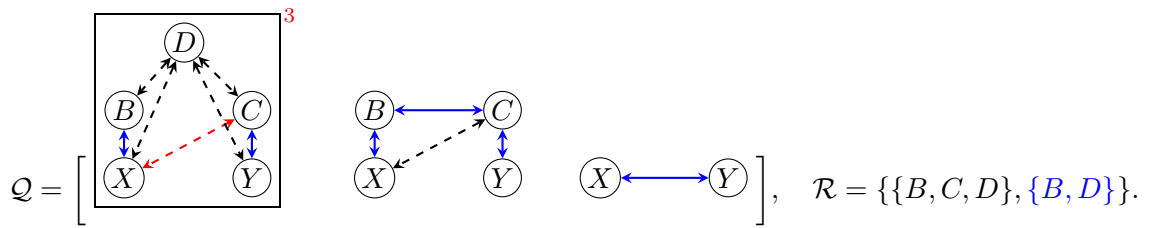
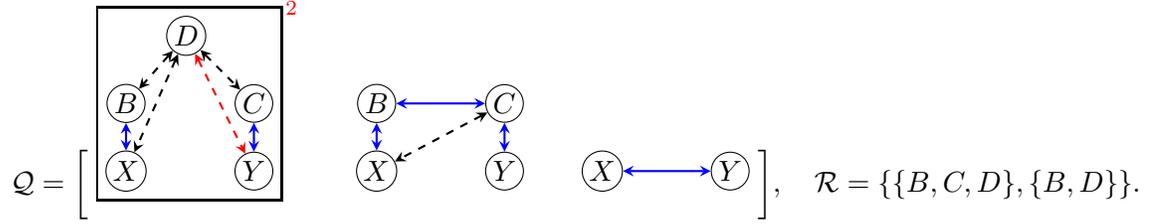
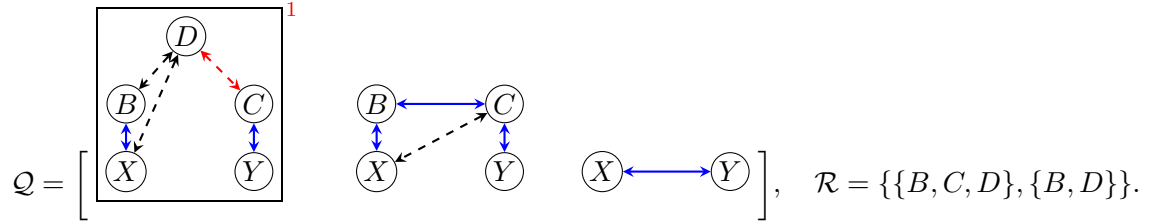
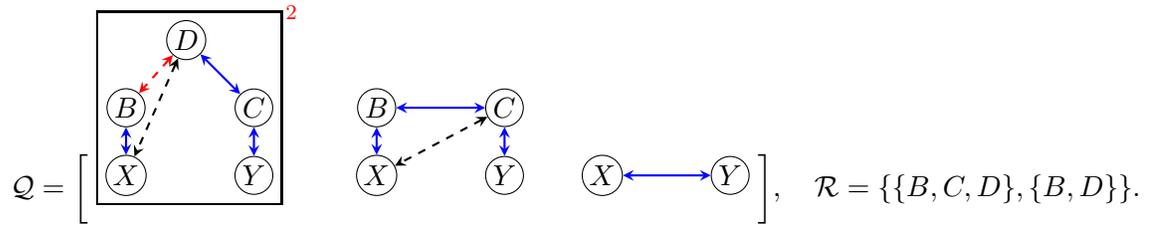
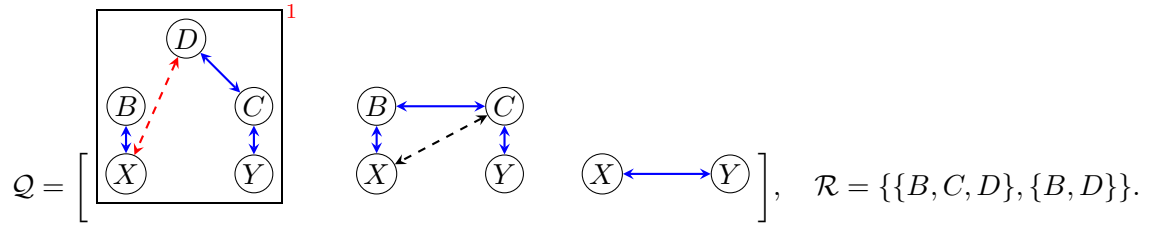
\begin{align*}
\mathcal{Q} &= \Biggl[\; \boxed{%
\;\Biggr], \\
\mathcal{R}&=\{\{B,C,D \}, \{B, D\} \}.
\end{align*}

\item Further, due to $\textnot B \confarc X \mid C, D$, we have
\begin{align*}
\mathcal{Q} &= \Biggl[\; \boxed{%
\;\Biggr], \\
\mathcal{R}&=\{\{B,C,D \}, \{B, D\} \}.
\end{align*}

\item  Hence, $\{B,C,D\}$ is identified again as a sufficient adjustment set. This leaves $\mathcal{R}$ unchanged:
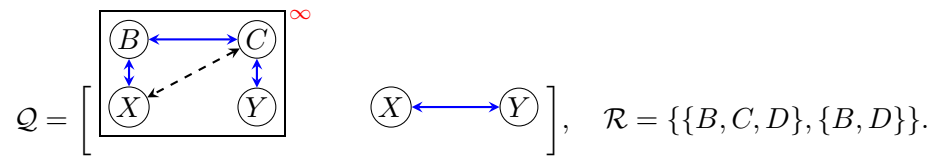
\begin{align*}
\mathcal{Q} &= \Biggl[\; \boxed{%
}^{\,\color{red} \infty} \;\Biggr],\\
\mathcal{R}&=\{\{B,C,D \}, \{B, D\} \}.
\end{align*}

\item Upon $\mathcal{Q} = []$, the algorithm terminates and returns
\[\mathcal{R}=\{\{B,C,D \}, \; \{B, D\}\},\]
which contains the only minimal sufficient adjustment set $\{B,D\}$.
\end{enumerate}
\clearpage{}

\end{document}